%% file: main-arxiv.tex
\documentclass[sigconf, nonacm]{acmart}

\AtBeginDocument{%
  }

\setcopyright{none}
\copyrightyear{}
\acmYear{}
\acmDOI{}

\acmPrice{}
\acmISBN{}

\usepackage{amsfonts}
\usepackage{amsthm}
\usepackage{tabularx}
\usepackage{hyperref}
\usepackage{graphicx, subcaption}
\usepackage{booktabs}

\include{preamble}

\begin{document}

\title{Faster Approximation Algorithms for Parameterized Graph Clustering and Edge Labeling}



\author{Vedangi Bengali}
\email{vedangibengali@tamu.edu}
\affiliation{%
	\institution{Texas A\&M University}
	\city{College Station}
	\state{Texas}
	\country{USA}
}

\author{Nate Veldt}
\email{nveldt@tamu.edu}
\affiliation{%
	\institution{Texas A\&M University}
	\city{College Station}
	\state{Texas}
	\country{USA}
}

\renewcommand{\shortauthors}{Vedangi Bengali and Nate Veldt}

\input{arxiv/sections/abstract}



\maketitle

\input{arxiv/sections/introduction-v2}
\input{arxiv/sections/prelims-v2}
\input{arxiv/sections/lambdastc}

\input{arxiv/sections/lamcc-algs}
\input{arxiv/sections/experiments}

\section{Conclusion}
We present new approximation algorithms for LambdaCC graph clustering framework that are far more scalable than existing approximation algorithms relying on LP relaxations with $O(n^3)$ constraints. 
We introduce the first combinatorial algorithm for LambdaCC in the parameter regime $\lambda \in (\frac12,1)$---where the problem interpolates between cluster editing and cluster deletion---which comes with a 6-approximation guarantee.  We then provide algorithms for all parameter regimes based on rounding a less expensive LP relaxation. A major theoretical benefit of these alternative LPs is that they are covering LPs. This means that the multiplicative weights update method provides fast combinatorial methods for finding approximation solutions. Although in our work we focused on using existing optimization software to exactly solve these relaxations, a clear direction for future research is to implement these faster approximate solvers in order to achieve additional runtime improvements. Another direction for future work is to see whether it is possible to develop a 3-approximation for all $\lambda \in (\frac12,1)$ by rounding the LambdaSTC LP. Though our theoretical approximation factors get increasingly worse for $\lambda \rightarrow 1$, in practice we see no deterioration in approximations. 
Finally, a compelling open question is whether we can develop an $O(\log n)$ approximation algorithm for LambdaCC that applies for all values of $\lambda$ and can be made purely combinatorial, and does not rely on the canonical LP.




\bibliographystyle{ACM-Reference-Format}
\bibliography{refs,cc-cd}


\end{document}

%% file: preamble.tex
\usepackage{amsfonts}
\usepackage{enumitem}
\usepackage{algorithm}
\usepackage[noend]{algorithmic}
\usepackage{mathtools}
\usepackage{graphicx} 

\newcommand{\lamcc}{\textbf{CC}_{(\lambda)}}
\newcommand{\lamstc}{\textbf{STC}_{(\lambda)}}
\newcommand{\minstc}{\textsc{minSTC}}
\newcommand{\minstcp}{\textsc{minSTC+}}

\usepackage{amsthm}

\newtheorem{theorem}{Theorem}
\numberwithin{theorem}{section}

\newtheorem{corollary}[theorem]{Corollary}

\numberwithin{observation}{section}

\numberwithin{definition}{section}

\numberwithin{conjecture}{section}

%% file: arxiv/sections/abstract.tex
\begin{abstract}
      Graph clustering is a fundamental task in network analysis where the goal is to detect sets of nodes that are well-connected to each other but sparsely connected to the rest of the graph. We present faster approximation algorithms for an NP-hard parameterized clustering framework called LambdaCC, which is governed by a tunable resolution parameter and generalizes many other clustering objectives such as modularity, sparsest cut, and cluster deletion. Previous LambdaCC algorithms are either heuristics with no approximation guarantees, or computationally expensive approximation algorithms. We provide fast new approximation algorithms that can be made purely combinatorial. These rely on a new parameterized edge labeling problem we introduce that generalizes previous edge labeling problems that are based on the principle of strong triadic closure and are of independent interest in social network analysis. Our methods are orders of magnitude more scalable than previous approximation algorithms and our lower bounds allow us to obtain a posteriori approximation guarantees for previous heuristics that have no approximation guarantees of their own.
\end{abstract}

%% file: arxiv/sections/introduction-v2.tex
\section{Introduction}

In network analysis, graph clustering is the task of partitioning a graph into well-connected sets of nodes (called communities, clusters, or modules), that are more densely connected to each other than they are to the rest of the graph~\cite{fortunato2016communitydetection,schaeffer2007graphclustering,porter2009communities}. 
This fundamental task has widespread applications across numerous domains, including detecting related genes in biological networks~\cite{shamir2004cluster,Ben-DorShamirYakhini1999}, finding communities in social networks~\cite{veldt2018correlation,newman2004modularity}, and image segmentation~\cite{ShiMalik2000}, to name only a few. A standard approach for finding clusters in a graph is to optimize some type of combinatorial objective function that encodes the quality of a clustering of nodes. Just as there are many different applications and reasons why one may wish to partition the nodes of a graph into clusters, there are many different types of objective functions for graph clustering~\cite{newman2004modularity,ShiMalik2000,shamir2004cluster,bohlin2014community,DelvenneYalirakiBarahona2010}, all of which strike a different balance between the goal of making clusters dense internally and the goal of ensuring that few edges cross cluster boundaries. In order to capture many different notions of community structure within the same framework, many graph clustering optimization objectives come with tunable \emph{resolution parameters}~\cite{schaub2012encoding,Veldt2019learning,ReichardtBornholdt2006,DelvenneYalirakiBarahona2010,newman2016equivalence}, which control the tradeoff between the internal edge density and the inter-cluster edge density resulting from optimizing the objective.

%
%
%
%

One of the biggest challenges in graph clustering is that the vast majority of clustering objectives are NP-hard. Thus, while it is often easy to define a new way to measure clustering structure, it is very hard to find optimal (or even certifiably near-optimal) clusters in practice for any given objective. There has been extensive theoretical research on approximation algorithms for different clustering objectives~\cite{arora2009scut,LeightonRao1999,veldt2018correlation,CharikarGuruswamiWirth2003}, but most of these come with high computational costs and memory constraints, often because they rely on expensive convex relaxations 
of the NP-hard clustering objective. On the other hand, scalable graph clustering algorithms have been designed based on local node moves and greedy heuristics~\cite{newman2004modularity,newman2006modularity,blondel2008fast,traag2019louvain,veldt2018correlation,shi2021scalable}, but these come with no theoretical approximation guarantees. As a result, it can be challenging to tell whether the structure of an output clustering depends more on the underlying objective function or on the mechanisms of the algorithm being used.

This paper focuses on an existing optimization graph clustering framework called LambdaCC~\cite{veldt2018correlation,Veldt2018ccgen,shi2021scalable,gan2020graph}, which comes with two key benefits. The first is that it can detect different types of community structures by tuning a resolution parameter $\lambda \in (0,1)$. Many existing clustering objectives can be recovered as special cases for specific choices of $\lambda$~\cite{veldt2018correlation}. The second benefit is that LambdaCC can be viewed as a special case of correlation clustering~\cite{BansalBlumChawla2004}, a framework for clustering based on similarity and dissimilarity scores, that has been studied extensively from the perspective of approximation algorithms~\cite{CharikarGuruswamiWirth2005,DemaineEmanuelFiatEtAl2006,Veldt2018ccgen}. As a result, LambdaCC directly inherits an $O(\log n)$ approximation algorithm that holds for any correlation clustering problem~\cite{CharikarGuruswamiWirth2005,DemaineEmanuelFiatEtAl2006} and is amenable to even better approximation guarantees in some parameter regimes. Gleich et al.~\cite{Veldt2018ccgen} showed that for very small values of $\lambda$, the $O(\log n)$ approximation is the best that can be achieved by rounding a linear programming relaxation (the most successful approach known for approximating the objective). However, a 3-approximation algorithm has been developed for the regime where $\lambda \geq 1/2$~\cite{veldt2018correlation}. Despite these results, LambdaCC suffers from a similar theory-practice gap as many other clustering frameworks. These previous approximation algorithms rely on expensive linear programming relaxations and are therefore not scalable. While faster heuristic algorithms do exist~\cite{veldt2018correlation,shi2021scalable}, these come with no approximation guarantees. 



\paragraph{The present work: fast approximation algorithms for parameterized graph clustering.}
We develop algorithms for LambdaCC that come with rigorous approximation guarantees and are also far more scalable than existing approximation algorithms for this problem. We present new algorithms for all values of the parameter $\lambda \in (0,1)$, focusing especially on the regime $\lambda \in (\frac12,1)$, since constant factor approximations are possible in this regime and have been a focus in previous research. This is also the regime where LambdaCC interpolates between two existing objectives known as cluster editing and cluster deletion~\cite{shamir2004cluster}. We first of all design a fast combinatorial approximation algorithm that returns a 6-approximation for any value of $\lambda \in (\frac12, 1)$, that runs in only $O(\sum d_v^2)$ time, where $d_v$ is the degree of node $v$. While this is a factor of 2 worse than the best existing 3-approximation, it is orders of magnitude faster than this previous approach, which requires solving an LP relaxation with $O(n^3)$ variables for an $n$-node graph and takes $\Omega(n^6)$ time. Our second algorithm is an improved $7-{2}/{\lambda}$ approximation for $\lambda \geq \frac12$ (which ranges from 3 to 5 as $\lambda \rightarrow 1$) based on rounding an LP relaxation with far fewer constraints. In numerical experiments, we confirm for a large collection of real-world networks that the number of constraints in this cheaper LP tends to be orders of magnitude smaller than the $O(n^3)$ constraint set of the canonical LP relaxation. It can also be run on graphs that are so large that even forming the $O(n^3)$ constraint matrix for the canonical LP relaxation leads to memory issues. Even more significantly, this cheaper LP that we consider is a \emph{covering} LP, a special type of LP that can be solved using combinatorial algorithms based on the multiplicative weights update method~\cite{fleischer2004fast,quanrud2020nearly,garg2004fractional}. 

We also adapt our techniques to obtain a $(1 + 1/\lambda)$ approximation by rounding the cheaper LP when $\lambda < \frac12$. As is the case when rounding the righter and more expensive canonical LP relaxation, this gets increasingly worse as $\lambda$ decreases. This is not surprising, given that even the canonical LP relaxation has an $O(\log n)$ integrality gap~\cite{Veldt2018ccgen}. Our $(1+1/\lambda)$ approximation is in fact quite close to the previous $1/\lambda$ approximation for small $\lambda$ that was previously developed by Gleich et al.~\cite{Veldt2018ccgen} based on the canonical LP.

All of our approximation algorithms rely on a new connection between LambdaCC and an edge labeling problem that is based on the social network analysis principle of strong triadic closure~\cite{easley2010networks,sintos2014using,granovetter1973strength}. This principle posits that if two people share strong links to a mutual friend, then they are likely to share at least a weak connection with each other. This principle has inspired a line of research on strong triadic closure (STC) labeling problems~\cite{sintos2014using,oettershagen2023inferring,veldt2022stc,gruttemeier2020relation,gruttemeierstrong}, which label edges in a graph as weak or strong (or in some cases add ``missing" edges) in order to satisfy the strong triadic closure property. Previous research has shown that unweighted variants of this labeling problem are related to cluster editing and cluster deletion~\cite{gruttemeier2020relation,gruttemeierstrong} (special cases of LambdaCC when $\lambda = 1/2$ and $\lambda \approx 1$ respectively). Recently it was shown that lower bounds and algorithms for these unweighted STC problems can be useful tools in designing faster approximation algorithms for cluster editing and cluster deletion~\cite{veldt2022stc}. We generalize this strategy by defining a new parameterized edge-labeling problem we call LambdaSTC, which provides new types of lower bounds for LambdaCC. We also provide a 3-approximation algorithm for LambdaSTC that applies for every value of $\lambda \in (0,1)$. All of these constitute new results for an edge labeling problem of independent interest in social network analysis, but our primary motivation is to use them to develop faster clustering approximation algorithms.

We demonstrate in numerical experiments that our algorithms are fast and effective, far surpassing their theoretical guarantees. In our experiments, we even find that solving our cheaper LP relaxation actually tends to return a solution that can quickly be certified to be the optimal solution for the more expensive canonical LP relaxation for LambdaCC. When this happens, we can use previous rounding techniques that guarantee a 3-approximation for $\lambda \geq \frac12$.


%% file: arxiv/sections/prelims-v2.tex
\section{Preliminaries and Related Work}
We begin with technical preliminaries on graph clustering, correlation clustering, and strong triadic closure edge labeling problems.

\subsection{The LambdaCC Framework}
Given an undirected graph $G = (V,E)$ the high-level goal of a graph clustering algorithm is to partition the node set $V$ into disjoint clusters in such a way that many edges are contained inside clusters, and few edges cross between clusters. These two goals are often in competition with each other, and there have been many different approaches for defining and forming clusters, all of which implicitly strike a different type of tradeoff between these goals. The LambdaCC clustering objective~\cite{veldt2018correlation} provides one approach for implicitly controlling this tradeoff using a \emph{resolution} parameter $\lambda \in (0,1)$. Formally, given $G = (V, E)$ and parameter $\lambda \in (0,1)$, LambdaCC seeks a clustering that minimizes the following objective
\begin{equation}
\label{lambdacc}
	\min_{\delta} \quad \sum_{(i,j)\in E}(1-\lambda)(1-\delta_{ij}) + \sum_{(i,j)\notin E}\lambda\delta_{ij},
\end{equation}
 where $\delta_{ij}$ is a binary cluster indicator for every node-pair, i.e., $\delta_{ij}=1$ if $i$ and $j$ are clustered together and $0$ otherwise. The number of clusters to form is not specified. Rather, the optimal number of clusters is controlled implicitly by tuning $\lambda$. Observe that the two pieces of the LambdaCC objective directly correspond to the two goals of graph clustering: the term $(1-\lambda)(1-\delta_{ij})$ for $(i,j) \in E$ is a penalty incurred if $i$ and $j$ are separated, and the term $\lambda \delta_{ij}$ for $(i,j) \notin E$ places a penalty on putting $i$ and $j$ together if they do not share an edge. The relative importance of the two competing goals of graph clustering (form clusters that are internally dense and externally sparse) is then controlled by tuning $\lambda$. Smaller values of $\lambda$ tend to produce a smaller number of (larger) clusters, and choosing large $\lambda$ leads to a larger number of (smaller) clusters. 

One of the benefits of the LambdaCC framework is that it generalizes and unifies a number of previously studied objectives for graph clustering, including the sparsest cut objective~\cite{arora2009scut}, cluster editing~\cite{shamir2004cluster, BansalBlumChawla2004}, and cluster deletion~\cite{shamir2004cluster}. It also is equivalent to the popular modularity clustering objective~\cite{newman2004modularity} under certain conditions. The definition of modularity depends on underlying null distributions for graphs. When the Erd\H{o}s-R\'{e}nyi null model is used, modularity is equivalent to Objective~\eqref{lambdacc} for an appropriate choice of $\lambda$. When the Chung-Lu null model is used, modularity can be viewed as a special case of a degree-weighted version of LambdaCC~\cite{veldt2018correlation}, though we focus on Objective~\eqref{lambdacc} in this paper. Finally, for an appropriate choice of $\lambda$, the LambdaCC objective is equivalent to graph clustering based on maximum likelihood inference for the popular stochastic block model~\cite{abbe2018community}, which can be seen from LambdaCC's relationship to modularity~\cite{newman2016equivalence}.

\subsection{Correlation Clustering}
The CC in LambdaCC stands for correlation clustering~\cite{BansalBlumChawla2004}, a framework for clustering based on pairwise similarity and dissimilarity scores. In the most general setting, an instance of weighted correlation clustering is given by a set of vertices $V$, along with two non-negative weights $(w_{ij}^+,w_{ij}^-)$ for each pair of distinct vertices $i,j \in V$. If nodes $i$ and $j$ are placed in the same cluster, this incurs a \emph{disagreement} penalty of $w_{ij}^-$ whereas if they are separated, a \emph{disagreement} penalty of $w_{ij}^+$ is imposed. In correlation clustering, disagreements are also called \emph{mistakes}. This terminology is especially natural when only one of the weights $(w_{ij}^+, w_{ij}^-)$ is positive and the other is zero (which is true for the most widely-studied special cases). In this case, each node pair $(i,j)$ is either ``similar'' ($w_{ij}^+ > 0$) and wants to be clustered together, or ``dissimilar'' ($w_{ij}^- > 0$) and wants to be clustered apart. A \emph{mistake} or \emph{disagreement} happens precisely when nodes are clustered in a way that does not match this ``preference.'' Formally, the \emph{disagreements minimization} objective for correlation clustering can be represented as the following integer linear program (ILP):
\begin{equation}
	\begin{aligned} \label{eq:cc_ilp}
		\text{ min}\quad &  {\textstyle \sum_{i,j}w_{ij}^+x_{ij} + w_{ij}^-(1-x_{ij})}\\
		\text{subject to} \quad &  x_{ij}\leq x_{ik}+x_{jk} \text{ for all } i,j,k\\
		& x_{ij}\in \{0,1\} \text{ for all } i,j   
	\end{aligned}
\end{equation}
where $x_{ij}$ is a binary distance variable between nodes $i,j$, i.e., $x_{ij} = 0$ means nodes $i$ and $j$ are clustered together, and $x_{ij} = 1$ means they are separated. The most well-studied special case is when $(w_{ij}^+, w_{ij}^-) \in \{(1,0), (0,1)\}$. This is known as complete unweighted correlation clustering, as it is often viewed as a clustering objective on a complete \emph{signed} graph where each pair of nodes either defines a positive edge or a negative edge. This is equivalent to the cluster editing problem~\cite{shamir2004cluster}, which seeks to add or delete the minimum number of edges in an unsigned graph $G = (V,E)$ to partition it into a disjoint set of cliques. This is in turn related to cluster \emph{deletion}, where one can only \emph{delete} edges in $G$ in order to partition it into cliques. Cluster deletion is the same as solving Objective~\eqref{eq:cc_ilp} when $(w_{ij}^+, w_{ij}^-) \in \{(1,0), (0,\infty)\}$ for every pair of nodes $(i,j)$.

\textbf{Approximation algorithms.} Correlation clustering is NP-hard even for special cases of cluster editing and cluster deletion, but many approximation algorithms have been designed~\cite{AilonCharikarNewman2008,CharikarGuruswamiWirth2005,ChawlaMakarychevSchrammEtAl2015}. Most of these algorithms rely on solving and rounding a linear programming (LP) relaxation of ILP~\eqref{eq:cc_ilp}, obtained by replacing $x_{ij} \in \{0,1\}$ with the constraint $x_{ij} \in [0,1]$. For the general weighted case, the best approximation guarantee is $O(\log n)$, which matches the integrality gap of the linear program~\cite{CharikarGuruswamiWirth2005,DemaineEmanuelFiatEtAl2006}. However, constant factor approximations are possible for certain weighted cases~\cite{AilonCharikarNewman2008,veldt2020parameterized,Ailon2011bcc}. Ailon et al.~\cite{AilonCharikarNewman2008} designed a fast randomized combinatorial algorithm called \textsc{Pivot} for the complete unweighted case. This algorithm repeatedly selects a uniform random pivot node in each iteration and clusters it together with its unclustered neighboring nodes that share a positive edge. This algorithm comes with an expected 3-approximation guarantee. However, for general weighted correlation clustering, it can produce poor results.

\textbf{Deterministic pivot.} A derandomized version of the standard \textsc{Pivot} algorithm was developed by van Zuylen and Williamson~\cite{vanzuylen2009deterministic}, which can be applied to a broader class of weighted correlation clustering problems.
Instead of randomly choosing pivot nodes, this technique relies on solving the LP relaxation of correlation clustering, constructing a derived graph $\hat{G}$ based on the solution to this LP, and then running a pivoting procedure in $\hat{G}$ that deterministically selects pivot nodes based on the LP output. They showed that this can produce a deterministic 3-approximation algorithm for the complete unweighted case, and can also be applied to other weighted cases including the case of probability constraints (where $w_{ij}^+ + w_{ij}^- = 1$ for every pair $\{i,j\}$). In proving these results, van Zuylen and Williamson presented a useful theorem (Theorem 3.1 in~\cite{vanzuylen2009deterministic}) that can be used as a general strategy for developing approximation algorithms for other special weighted variants. We state a version of this theorem below, as it will be a useful step for some of our results. The original theorem includes details for choosing pivot nodes deterministically. The approximation holds in expectation when choosing pivot nodes uniformly at random.
\begin{theorem}\label{thm_2.1}
 Consider an instance of weighted correlation clustering given by a node-set $V$ and weights $\{w^+_{ij},w^-_{ij}\}_{i\neq j}$. Let $b_{ij}$ represent a budget for node pair $\{i,j\}$ where $i\neq j$, and $\hat{G}= (V,\hat{E})$ be a graph which for $\alpha>0$ satisfies the following conditions:
	\begin{enumerate} 
		\item $w_{ij}^- \leq \alpha b_{ij}$ for all edges $(i,j) \in \hat{E}$, and \\
		$w_{ij}^+ \leq \alpha b_{ij}$ for all edges $(i,j) \notin \hat{E}$,
		\item $w_{ij}^+ + w_{jk}^+ + w_{ik}^- \leq \alpha (b_{ij}+b_{jk}+b_{ik})$ for every triplet $\{i,j,k\}$ in $\hat{G}$ where $\{i,j\}\in \hat{E}, \{j,k\}\in \hat{E}, \{i,k\}\notin \hat{E}$.
	\end{enumerate}
	Applying \textsc{Pivot} to $\hat{G}$ will return a clustering solution with an expected weight of disagreements bounded above by $\alpha\sum_{i<j}b_{ij}$.
\end{theorem}

\textbf{Approximations for LambdaCC.} The LambdaCC objective on a graph $G = (V,E)$ corresponds to a special case of Objective~\eqref{eq:cc_ilp} where $(w_{ij}^+,w_{ij}^-) = (1-\lambda,0)$ if $(i,j) \in E$ and $(w_{ij}^+,w_{ij}^-) = (0,\lambda)$ if $(i,j) \notin E$. Veldt et al.~\cite{veldt2018correlation} previously showed a 3-approximation algorithm for the case where $\lambda \geq 1/2$, based on rounding the standard correlation clustering LP relaxation.
However, because this LP has $O(n^3)$ constraints for an $n$-node graph, in practice it is challenging to solve it for graphs with even a thousand nodes. Later, Gleich et al.~\cite{Veldt2018ccgen} proved that the LP integrality gap can be $O(\log n)$ for small values of $\lambda$. They also developed approximation guarantees for smaller values $\lambda$, but these get increasingly worse as $\lambda \rightarrow 0$. Faster heuristics algorithms for LambdaCC have also been developed~\cite{veldt2018correlation,shi2021scalable}, but these come with no approximation guarantees. Thus, a limitation of this previous work is that existing LambdaCC algorithms either depend on an extremely expensive linear programming relaxation or come with no guarantees. The focus of our paper is to bridge this gap.

\subsection{Strong Triadic Closure and Edge Labeling}
In social network analysis, the principle of strong triadic closure \cite{easley2010networks,granovetter1973strength} posits that two individuals in a social network will share at least a weak connection if they both share strong connections to a common friend. This has been used to define certain types of edge labeling problems where the goal is to label edges in a graph $G = (V,E)$ in such a way that this principle holds~\cite{sintos2014using,oettershagen2023inferring, gruttemeier2020relation,gruttemeierstrong,adriaens2020relaxing}. 

Given a graph $G = (V,E)$ (which could represent an observed set of social interactions), a triplet of vertices $(i,j,k)$ is an \textbf{open wedge} centered at $j$ if the vertex pairs $(i,j)$ and $(j,k)$ are edges (i.e., in $E$) while $(i,k)$ is not. The strong triadic closure principle suggests that if such an open wedge exists, then either $(i,j)$ or $(j,k)$ is a \emph{weak} edge, or else $(i,k)$ is a \emph{missing connection} that should appear as an edge in $G$ but was not observed when $G$ was constructed. With this principle in mind, Sintos and Tsaparas~\cite{sintos2014using} defined the strong triadic closure labeling problem (\textsc{minSTC}), where the goal is to label edges as \emph{weak} and \emph{strong} so that every open wedge contains as least one \emph{weak} edge, and in such a way that the number of edges labeled \emph{weak} is as small as possible. They showed that the problem is NP-hard but has a 2-approximation algorithm based on reduction to the \textsc{Vertex Cover} problem. They also considered a variation of the problem that allows for edge additions (\textsc{minSTC+}).

In our paper, we use $\mathcal{W}_j$ to denote the set of open wedges centered at $j$ in $G$, and let $\mathcal{W} = \cup_{j \in V} \mathcal{W}_j$. We use the term STC-labeling to indicate labeling of node pairs $(i,j) \in {V \choose 2}$ that satisfies the strong triadic closure in the following sense: for every open wedge $(i,j,k)$ centered at $j$, at least one of the edges $\{(i,j), (j,k)\}$ is labeled \emph{weak}, \emph{or} the non-edge $(i,j) \notin E$ is labeled as a \emph{missing edge}. Such labeling is encoded by a collection of weak edges denoted as $E_\mathit{weak}\subseteq E$, as well as a set of missing edges denoted as $E_\mathit{miss}\subseteq ({V \choose 2}-E)$. 
The \textsc{minSTC+} problem seeks an STC-labeling $(E_\mathit{weak}, E_\mathit{miss})$ that minimizes $|E_\mathit{weak}| + |E_\mathit{miss}|$. This can be formally cast as the following ILP:
\begin{equation}
	\label{minstc+}
	\begin{array}{lll}
		\min & \displaystyle{\sum_{(i,j) \in {V \choose 2}}} z_{ij} &\\
		\text{s.t. } & z_{ij} + z_{jk} + z_{ik} \geq 1 &\text{ if $(i,j,k) \in \mathcal{W}_j$}\\
		&  z_{ij} \in \{0,1\} &\text{ for $(i,j) \in {V \choose 2}$}.
	\end{array}		
\end{equation}
If $z_{ij}=1$, this represents the presence of either a weak edge (if $(i,j) \in E$) or a missing edge (if $(i,j)\notin E$). This problem is also NP-hard but can be reduced to \textsc{Vertex Cover} in a 3-uniform hypergraph in order to obtain a 3-approximation algorithm. 

While the strong triadic closure principle and the resulting edge labeling problems are of their own independent interest, we are particularly interested in these problems given their relationships with certain clustering objectives. The solution for \textsc{minSTC} is known to lower bound the cluster deletion objective, and \textsc{minSTC+} similarly lower bounds cluster editing~\cite{gruttemeier2020relation,gruttemeierstrong,konstantinidis2018strong,veldt2022stc}. The LP relaxations for these problems, therefore, provide lower bounds for cluster deletion and clustering editing that are cheaper and easier to compute than the standard linear programming relaxations. Veldt~\cite{veldt2022stc} recently showed how to round these LP relaxations---and how to round approximation solutions for \textsc{minSTC+} and \textsc{minSTC}---to develop faster approximation algorithms for cluster editing and cluster deletion.
We generalize these techniques in order to develop faster approximation algorithms for the full parameter regime of LambdaCC.

%% file: arxiv/sections/lambdastc.tex
\section{Lambda STC Labeling}
We now introduce a parameterized edge labeling problem called LambdaSTC, which generalizes previous edge labeling problems and can also be used to develop new approximations for LambdaCC.


\textbf{Problem definition.} Given graph $G=(V,E)$ and a parameter $\lambda$, 
LambdaSTC is the problem of finding an STC-labeling $(E_\mathit{weak}, E_\mathit{miss})$ that minimizes $(1-\lambda)|E_\mathit{weak}| + \lambda|E_\mathit{miss}|$. This can be formulated as:
\begin{equation}
	\label{eq:lamstc} 
	\begin{aligned}
		\text{ min}\quad &  {\textstyle \sum_{(i,j)\in E}(1-\lambda)z_{ij} + \sum_{(i,j)\notin E}\lambda z_{ij}}\\
		\text{s.t} \quad &  z_{ij}+ z_{jk}+z_{ik}\ge 1 \text{  for all } (i,j,k) \in \mathcal{W}_j\\
		& z_{ij} \in \{0,1\} \text{ for all } (i,j) \in {\textstyle {V \choose 2}.   }
	\end{aligned}
\end{equation}
We first note that this problem is equivalent to the \textsc{minSTC+} problem when $\lambda = \frac12$. When $\lambda$ is close enough to 1, LambdaSTC is equivalent to \textsc{minSTC}. To see why, note that if $\lambda > |E| (1-\lambda)$, then labeling a single non-edge as ``missing'' is more expensive than labeling \emph{all} edges in $E$ as ``weak''. Hence, with a couple of steps of algebra, we can see that when $\lambda > \frac{|E|}{|E| + 1}$, the optimal LambdaSTC solution will not place any non-edges in $E_\mathit{miss}$, but will only add edges to $E_\mathit{weak}$ in order to construct a valid STC-labeling, so this differs from  \textsc{minSTC} only by a multiplicative constant factor. 

Varying $\lambda$ between $\frac12$ and $1$ offers us the flexibility to interpolate between \textsc{minSTC+} and \textsc{minSTC}. Meanwhile, the $\lambda < \frac12$ regime corresponds to a new family of edge labeling problems where it is cheaper to label non-edges as missing. 
If we think of the graph $G = (V,E)$ as a (potentially noisy) representation of some social network observed from the real world, then the parameter $\lambda$ can be chosen based on a user's belief about the accuracy of the process that was used to observe edges. If the user has a strong belief that there are many friendships in the social network that were just not directly observed (and hence are not included as edges in the graph) then a smaller value of $\lambda$ may be appropriate. If missing edges are unlikely, then a large value of $\lambda$ is appropriate.

\textbf{Approximating LambdaSTC.}
Approximation algorithms for $\minstc$ and $\minstcp$ can be obtained by reducing these problems to \emph{unweighted} \textsc{Vertex Cover} problems (in graphs and hypergraphs, respectively)~\cite{sintos2014using,gruttemeierstrong}. We generalize this approach and design an approximation algorithm that applies to LambdaSTC for any choice of the parameter $\lambda$, based on the Local-Ratio algorithm for weighted \textsc{Vertex Cover}~\cite{bar1985local}. 
\begin{algorithm}[tb]
	\caption{$\textsc{CoverLabel}(G,\lambda)$}
	\label{alg:coverlabel}
	\begin{algorithmic}[1]
		\STATE{\bfseries Input:} Undirected graph $G = (V,E); \lambda \in (0,1)$
		\STATE {\bfseries Output:} LambdaSTC labeling $\{E_{weak},E_{miss}\}$ of $G$.
            \STATE $E_{miss},E_{weak} \leftarrow \emptyset $ // Initialize empty sets 
		\STATE  $\forall (i,j) \in E$ set $r_{ij} = (1-\lambda)$ ;  $\forall (i,j) \notin E$ set $r_{ij} = \lambda$ 
		\FOR{$\{i,j,k\} \in \mathbf{W}_j$}
            \STATE M = $\min\{r_{ij},r_{jk},r_{ik}\}$
            \STATE $r_{ij} \leftarrow r_{ij} - M$; \quad $r_{jk} \leftarrow r_{jk} - M$; \quad $r_{ik} \leftarrow r_{ik} - M$
            \ENDFOR
	
		\STATE $E_{miss} = \{ (i,j) \notin E \colon r_{ij}=0 \}$;  $E_{weak} = \{ (i,j) \in E \colon r_{ij}=0 \}$
		
		\STATE Return $\{E_{miss},E_{weak}\}$
	\end{algorithmic}
\end{algorithm}
Algorithm~\ref{alg:coverlabel} is pseudocode for our method, \textsc{CoverLabel}. This method ``covers'' all the open wedges in graph $G$ by either adding a missing edge between a pair of non-adjacent nodes or labeling at least one of the two edges as weak. This can be seen as finding a weighted vertex cover on a 3-uniform hypergraph $\mathcal{H} = (\mathcal{V,E})$ constructed as follows:
\begin{itemize}
    \item Every node pair $(i,j)\in {V \choose 2} $ is assigned to a vertex $v_{ij}$ in $\mathcal{V}$ with a node-weight of $(1-\lambda)$ if $(i,j)\in E$ and $\lambda$ otherwise.
    \item A hyperedge $\{v_{ij},v_{jk},v_{ik}\} \in \mathcal{E}$ is created for each open wedge $\{i,j,k\} \in \mathcal{W}$.
\end{itemize}
Nodes in $\mathcal{H}$ correspond to edges in $G$, and hyperedges in $\mathcal{H}$ correspond to open wedges in $G$. Therefore, a vertex cover in $\mathcal{H}$ corresponds to a labeling of edges in $G$ that ``covers'' all open wedges in a way that produces an STC-labeling. If the covered vertex is associated with an edge $(i,j) \in E$, we consider $(i,j)$ a weak edge. However, if it corresponds to a non-edge pair $(i,j) \notin E$, this is labeled as a missing edge. This provides an approximation-preserving reduction from LambdaSTC to 3-uniform hypergraph weighted \textsc{Vertex Cover}, so employing a 3-approximation algorithm for hypergraph vertex cover results in a 3-approximation for LambdaSTC. \textsc{CoverLabel} is equivalent to {implicitly} applying the Local-Ratio algorithm~\cite{bar1985local} to the hypergraph $\mathcal{H}$ described above. By \emph{implicitly}, we mean that we do not form $\mathcal{H}$ explicitly, but we apply the mechanics of this algorithm directly to find an STC-labeling in $G$. The following theorem follows from the guarantee of the Local-Ratio algorithm for node-weighted 3-uniform hypergraphs.
\begin{theorem}
Algorithm \ref{alg:coverlabel} is a 3-approximation algorithm for the LambdaSTC labeling problem.
\end{theorem}

\textbf{New Lower bounds for LambdaCC.} The LambdaSTC objective lower bounds LambdaCC, and a solution to the LambdaSTC LP yields a new type of lower bound for LambdaCC. To see why, consider the following change of variables: $z_{ij} = x_{ij}$ if $(i,j)\in E$, and $z_{ij} = 1-x_{ij}$ otherwise. This gives us the following equivalent formulation for the LP relaxation of LambdaSTC:
\begin{equation}
	\label{eq:cheapcc} 
	\begin{aligned}
		\text{ min}\quad &  {\textstyle \sum_{(i,j)\in E}(1-\lambda)x_{ij} + \sum_{(i,j)\notin E}\lambda(1-x_{ij})}\\
		\text{s.t.} \quad &  x_{ij}+ x_{jk} \ge x_{ik} \text{ if} (i,j,k) \in \mathcal{W}_j\\
		& 0\le x_{ij} \le 1 \text{ for all } (i,j) \in {\textstyle {V \choose 2}   }.
	\end{aligned}
\end{equation}
This linear program shares the same objective function as the canonical LambdaCC LP relaxation, but has a subset of the $O(n^3)$ triangle inequality constraints. In particular, it only constraints $x_{ij}+x_{jk}\ge x_{ik}$ when $\{i,j,k\}$ is an open wedge, rather than for all triplets of edges. This makes this LP relaxation easier to solve on a large scale. Furthermore, this is an example of a \emph{covering} LP, which can be solved much more quickly than a generic LP using the multiplicative weights update method~\cite{fleischer2004fast,quanrud2020nearly,garg2004fractional}.

%% file: arxiv/sections/lamcc-algs.tex
\section{Faster LambdaCC Algorithms}
We now present faster algorithms for LambdaCC, using  lower bounds derived from the LambdaSTC objective. 
For a fixed $\lambda$ value, we use the notation $\lamstc$ and $\lamcc$ to represent the optimal solution values for LambdaSTC and LambdaCC, respectively.

\subsection{CoverFlipPivot algorithm}
We present the first combinatorial algorithm for LambdaCC, called \textsc{CoverFlipPivot} (CFP), which provides a 6 approximation for every $\lambda \ge 1/2$. As outlined in Algorithm~\ref{alg:CFP}, CFP comprises three steps:
\begin{enumerate}
    \item \emph{Cover:} Generate a feasible LambdaSTC labeling of $G$ by running the 3-approximate \textsc{CoverLabel} algorithm.
    \item \emph{Flip:} Flip the edges in the original graph $G=(V, E)$ to create a derived graph $\hat{G}=(V, \hat{E})$, by deleting `weak' edges $E_\mathit{weak}$ and adding `missing' edges $E_\mathit{miss}$.
    \item \emph{Pivot:} Run \textsc{Pivot} on the derived graph $\hat{G}$.
\end{enumerate}

\begin{algorithm}[tb]
	\caption{$\textsc{CoverFlipPivot}(G,\lambda)$}\label{alg:CFP}
	\begin{algorithmic}[1]
		\STATE{\bfseries Input:} Undirected graph $G = (V,E); \lambda \ge 1/2$
		\STATE {\bfseries Output:} Feasible LambdaCC clustering of $G$.
  		\STATE \textit{Cover:} $\{E_\mathit{weak},E_\mathit{miss}\}$ = \textsc{CoverLabel}$(G,\lambda)$ 
  \STATE \textit{Flip:}  Construct $\hat{G} = (V, \hat{E})$ where $\hat{E} = E_\mathit{miss} \cup (E - E_{\mathit{weak}})$
	\STATE Return $\textsc{Pivot}(\hat{G})$
	\end{algorithmic}
	\label{alg:6ce}
\end{algorithm}

Before proving any approximation guarantees for CFP, we begin with a more general result that sheds light on the relationship between LambdaSTC and LambdaCC. This generalizes previous results showing that the optimal cluster deletion and \textsc{minSTC} objectives (and similarly, the cluster editing and \textsc{minSTC+} objectives) differ by at most a factor of 2~\cite{veldt2022stc}.
\begin{theorem}
    \label{thm:combi}
    Given an input graph $G=(V,E)$, a clustering parameter $\lambda \ge 1/2$, and an STC-labeling $(E_\mathit{weak},E_\mathit{miss})$, running \textsc{Pivot} on the derived graph $\hat{G} = (V,\hat{E})$ where $\hat{E} = E_\mathit{miss} \cup (E-E_\mathit{weak})$, returns a LambdaCC clustering solution with a bound of $2((1-\lambda)|E_\mathit{weak}|+\lambda|E_\mathit{miss}|)$ on the expected cost of disagreements.
\end{theorem}
\begin{proof}
    To prove this result, we show that all conditions of Theorem~\ref{thm_2.1} are satisfied for $\alpha = 2$. Recall that for the LambdaCC framework, weights are defined as:
	\begin{equation}
		\label{4owij}
		(w_{ij}^+, w_{ij}^-) = \begin{cases}
			(1-\lambda,0) & \text{ if $(i,j) \in E$ }\\
			(0,\lambda) & \text{ if $(i,j) \notin E$}.
		\end{cases}
	\end{equation}
To bound the LambdaCC objective in terms of the STC-labeling, we define budgets based on flipped edges:
	\begin{equation}
		\label{4ocij}
		b_{ij} = \begin{cases}
			(1-\lambda) & \text{ if $(i,j) \in E_\mathit{weak}$ }\\
			\lambda & \text{if $(i,j) \in E_\mathit{miss}$}\\
			0 & \text{ otherwise.}
		\end{cases}
	\end{equation}
The sum of budgets can now be written as $\sum_{i<j}b_{ij}=(1-\lambda)|E_\mathit{weak}| + \lambda|E_\mathit{miss}|$. Now, we show that Condition (1) of Theorem~\ref{thm_2.1} is satisfied for $\alpha=2$, by considering four cases:
 \begin{align*}
		&\text{ if }(i,j) \in \hat{E } \cap E \text{ then } w_{ij}^- = 0 \leq 2 b_{ij}, \\
		&\text{ if }(i,j) \in \hat{E} \text{ but } (i,j) \notin E \text{ then } w_{ij}^- = \lambda \le 2\lambda = 2 b_{ij}, \\
		&\text{ if }(i,j) \notin \hat{E} \text{ and } (i,j) \notin E \text{ then } w_{ij}^+ = 0 \leq 2 b_{ij} \text{ and }\\
		&\text{ if }(i,j) \notin \hat{E} \text{ but } (i,j) \in E \text{ then } w_{ij}^+ = 1-\lambda \leq 2(1-\lambda) = 2 b_{ij}.
	\end{align*}
Next, we check Condition (2) of Theorem~\ref{thm_2.1}, i.e., we prove that $w_{ij}^+ + w_{jk}^{+} + w_{ik}^- \leq \alpha(b_{ij}+b_{jk}+b_{ik})$ for every open wedge $(i,j,k)$ centered at $j$ in $\hat{G}$.
The budgets and weights $\{w_{ij}^+, w_{jk}^+, w_{ik}^-\}$ depend on which pairs $(i,j)$, $(j,k)$, $(i,k)$ are edges in $G$. Table~\ref{tab:cfpproof} covers all 8 cases for how a triplet of nodes in $G$ could be mapped to an open wedge centered at $j$ in $\hat{G}$. The first column indicates which of the pairs $\{(i,j)(j,k)(i,k)\}$ are edges in $G$, e.g., Y-Y-N (``yes-yes-no'') means that $(i,j)$ and $(j,k)$ are in $E$, but $(i,k)$ is not. In each case, we show why $w_{ij}^+ + w_{jk}^{+} + w_{ik}^-$ is a lower bound for $\alpha(b_{ij}+b_{jk}+b_{ik})$ when $\alpha=2$. By Theorem~\ref{thm_2.1} the total cost of mistakes is then bounded by $\alpha\sum_{i<j}b_{ij} = 2((1-\lambda)|E_\mathit{weak}|+\lambda|E_\mathit{miss}|)$. 
\begin{table}[htbp]
    \centering
    \caption{Proof of Theorem~\ref{thm_2.1} Condition (2)}
    \label{tab:cfpproof}
    \vspace{-\baselineskip}
    \begin{tabular}{cllc} 
        \toprule
        Edges in $E$ &  &  \\
        $ij-jk-ik$ & $2(b_{ij} + b_{jk} + b_{ik})$ & $w_{ij}^+ + w_{jk}^+ + w_{ik}^-$  \\
        \midrule
        Y-Y-Y & $2(0 + 0 + 1-\lambda) = 2(1-\lambda)$ &  $2(1-\lambda)$ \\
        Y-Y-N & Not Applicable & N.A \\
        Y-N-N & $2(0 + \lambda + 0) = 2\lambda$ & $1 \le 2\lambda$  \\
        Y-N-Y & $2(0 + \lambda + 1-\lambda) = 2$ & $1-\lambda < 2$ \\
        N-N-Y & $2(\lambda + \lambda + 1-\lambda) = 2(1+\lambda)$ & $0 < 2(1+\lambda)$ \\
        N-Y-Y & $2(\lambda + 0 + 1-\lambda) = 2$ & $1-\lambda < 2$ \\
        N-Y-N & $2(\lambda + 0 + 0) = 2\lambda$ & $1 \le 2\lambda$ \\
        N-N-N & $2(\lambda + \lambda + 0) = 2\lambda$ & $\lambda < 2\lambda$ \\
        \bottomrule
    \end{tabular}
\end{table}
The second row of Table~\ref{tab:cfpproof} corresponds to the case where an open wedge in $G$ maps to an open wedge in $\hat{G}$. However, this is in fact impossible, as it implies that none of the node pairs in $(i,j,k)$ were \emph{flipped}, even though $(i,j,k)$ is an open wedge in $G$, which violates the assumption that $(E_\mathit{weak}, E_{miss})$ is an STC-labeling.
\end{proof}

\begin{corollary}
Let $\mathcal{A}$ be a $\beta$-approximation algorithm for LambdaSTC and fix $\lambda \geq 1/2$. Running the procedure in Theorem~\ref{thm:combi} on the solution $(E_\mathit{weak}, E_\mathit{miss})$ obtained from $\mathcal{A}$ yields a $(2\beta)$-approximate solution for LambdaCC.
Thus, $\lamstc \leq \lamcc \leq 2 \lamstc$, and Algorithm \ref{alg:CFP} is a randomized 6-approximation for LambdaCC.
\end{corollary}
\begin{proof}
    The optimal LambdaCC solution provides an upper bound for the optimal LambdaSTC solution, i.e., $\lamstc \leq \lamcc$. Algorithm $\mathcal{A}$ produces a $\beta$-approximate labeling solution $(E_\mathit{weak},E_\mathit{miss})$ with the LambdaSTC objective value $\text{STC}({\mathcal{A}}) = ((1-\lambda)|E_\mathit{weak}|+\lambda|E_\mathit{miss}|)$, so we have that
%
    \begin{align*}
\lamstc \leq \text{STC}({\mathcal{A}}) \leq \beta \lamstc \leq \beta \lamcc.
\end{align*}
Therefore, $\frac{1}{\beta}\text{STC}({\mathcal{A}})$ lower bounds $\lamcc$. Applying Theorem~\ref{alg:CFP} with algorithm $\mathcal{A}$, we obtain a clustering with LambdaCC objective score of $\text{CC}({\mathcal{A}}) \leq 2((1-\lambda)|E_\mathit{weak}|+\lambda|E_\mathit{miss}|) = 2\text{STC}_{\mathcal{A}}$. Thus,
\begin{align*}
\text{CC}({\mathcal{A}}) \leq 2\cdot \text{STC}({\mathcal{A}}) \leq 2\cdot\beta \cdot\lamcc,
\end{align*}
so $\text{CC}({\mathcal{A}})$ is a $(2\beta)$-approximation for LambdaCC. If $\mathcal{A}$ optimally solves LambdaSTC, then $\beta = 1$ and so $\lamstc \leq \lamcc \leq 2\lamstc$. If $\mathcal{A}$ represents our 3-approximate \textsc{CoverLabel} algorithm for LambdaSTC, then combining it with Theorem~\ref{thm:combi} shows that Algorithm~\ref{alg:6ce} is a $2\cdot 3 = 6$-approximation for LambdaCC.
\end{proof}
Algorithm~\ref{alg:CFP} can be easily derandomized using the deterministic strategy for choosing pivot nodes for Theorem~\ref{thm_2.1} (see~\cite{vanzuylen2009deterministic}).
 
 
\subsection{Faster LP algorithm}
Algorithm~\ref{alg:labelcclp} is an approximation algorithm for LambdaCC based on rounding the LambdaSTC LP relaxation. This LP has $|\mathcal{W}|$ constraints, whereas the canonical LP has $O(n^3)$. In the worst case, it is possible for $|\mathcal{W}| = O(n^3)$, but our experimental results demonstrate that $|\mathcal{W}|$ is far smaller for all of the real-world graphs we consider. Even more significantly, the LambdaSTC LP is a \emph{covering} LP, which makes it possible to use fast existing techniques for approximating covering LPs. This leads to much faster algorithms, at the expense of only a slightly worse approximation factor since the LP is only solved approximately. The next section provides a more detailed runtime analysis.

Our approach for rounding the LambdaSTC LP follows a similar strategy as CFP, which involves building a new graph $\hat{G}$ and then running \textsc{Pivot}. The construction of $\hat{G}$ depends on the LP variables $\{x_{ij}\}$, the edge structure in $G$, and the value of $\lambda$. When $\lambda \geq 1/2$, we always ensure that a non-edge in $G$ maps to a non-edge in $\hat{G}$. For $\lambda < 1/2$, we always ensure that an edge in $G$ maps to an edge in $\hat{G}$. We first prove that the algorithm has an approximation factor that ranges from $3$ to $5$ as $\lambda$ goes from $1/2$ to $1$. 
\begin{algorithm}[tb]
\label{alg:labelcclp}
\caption{Rounding the LambdaSTC LP into a clustering}
\label{alg:labelcclp}
	\begin{algorithmic}[1]
		\STATE{\bfseries Input:} Undirected graph $G=(V,E) ; \lambda \in (0,1)$ 
		\STATE {\bfseries Output:} Clustering of $G$.
		\STATE Solve LambdaSTC LP~\eqref{eq:cheapcc} and obtain fractional $x_{ij}$ values
		\STATE Construct $\hat{G} = (V, \hat{E})$ where 
		\begin{equation*}
            \hat{E} =
            \begin{cases}
                \{(i,j): (i,j)\in E \textit{ and } x_{ij} < \frac{2\lambda}{7\lambda - 2}\} & \text{ if $\lambda \geq 1/2$} \\
                \{(i,j): (i,j)\in E \textit{ or } x_{ij} < \frac{\lambda}{1 + \lambda}\} & \text{ if $\lambda < 1/2$}
            \end{cases}
		\end{equation*}
		\STATE Return $\textsc{Pivot}(\hat{G})$
	\end{algorithmic}
\end{algorithm}
\begin{theorem}
	\label{thm:cheap}
	 When $\lambda \geq 1/2$, Algorithm~\ref{alg:labelcclp} is a randomized $(7-\frac{2}{\lambda})$-approximation algorithm for LambdaCC. 
\end{theorem}
\begin{proof}
To prove Theorem~\ref{thm:cheap}, we show the conditions in Theorem~\ref{thm_2.1} are satisfied for $\alpha = (7-\frac{2}{\lambda})$. In our analysis, we define budgets $b_{ij}$ for each distinct pair of nodes $(i,j)$ based on the LP objective~\eqref{eq:cheapcc}. Specifically, we set $b_{ij} = (1-\lambda)x_{ij}$ if $(i,j) \in E$, and $b_{ij} = \lambda(1-x_{ij})$ if $(i,j)\notin E$. 
We begin by checking Condition (1) in Theorem~\ref{thm_2.1} for each distinct pair of nodes $(i,j)$, i.e.,
    \begin{align}
    	\label{eq:cond1.1}
        w_{ij}^- \le \alpha b_{ij} \text{   for all } (i,j) \in \hat{E}\\
        \label{eq:cond1.2}
        w_{ij}^+ \le \alpha b_{ij} \text{   for all } (i,j) \notin \hat{E}.
    \end{align}
    If $(i,j)\in \hat{E}\cap E$, then $w_{ij}^-=0$ so Condition~\eqref{eq:cond1.1} holds. Similarly, when $(i,j)\notin \hat{E}$ and $(i,j) \notin E$, then $w_{ij}^+=0$, trivially satisfying Condition \eqref{eq:cond1.2}. By the construction of $\hat{G}$, it is impossible for $(i,j)\in\hat{E}$ if $(i,j) \notin E$. So the last case to consider is when $(i,j)\notin\hat{E}$ and $(i,j) \in E$, in which case 
$x_{ij}\ge \frac{2\lambda}{7\lambda - 2}$, so
    \begin{align*}
        w_{ij}^+ = (1-\lambda) < \Big(7-\frac{2}{\lambda}\Big)(1-\lambda)\Big(\frac{2\lambda}{7\lambda - 2} \Big) \leq \alpha(1-\lambda)x_{ij} = \alpha b_{ij}.
    \end{align*}
        
    Next, for every triplet $(i,j,k)$ such that $(i,j)\in \hat{E}$, $(j,k)\in \hat{E}$ and $(i,k)\notin \hat{E}$, we need to check that 
    \begin{equation}
    	\label{cond2}
    	w_{ij}^+ + w_{jk}^+ + w_{ik}^- \leq \alpha(b_{ij}+b_{jk}+b_{ik}).
    \end{equation}
	By our construction of $\hat{G}$, if $(i,j)$ and $(j,k) \in \hat{E}$, then $(i,j)$ and $(j,k)$ are both edges in $G$ as well. Since $(i,k) \notin \hat{E}$, $(i,k)$ may or may not be an edge in $G$, we prove~\eqref{cond2} by considering two cases.
 
    \noindent\textbf{Case 1:} $(i,k) \notin E$. Here we have $(b_{ij},b_{jk},b_{ik}) = ((1-\lambda)x_{ij},(1-\lambda)x_{jk},\lambda (1-x_{ik}))$ and $(w_{ij}^+,w_{jk}^+,w_{ik}^-) = (1-\lambda,1-\lambda,\lambda)$. Using the open wedge inequality $x_{ik}\le x_{ij}+x_{jk}$ and the fact that both $x_{ij},x_{jk} < \frac{2\lambda}{7\lambda - 2}$, we know $x_{ik} < \frac{4\lambda}{7\lambda - 2}$. Therefore,
	\begin{align*}
 	\alpha(b_{ij}+b_{jk}+b_{ik})  
  &=\alpha((1-\lambda)(x_{ij}+x_{jk}) +\lambda(1-x_{ik}))\\
  &\geq \alpha((1-\lambda)x_{ik} + \lambda -\lambda x_{ik}) \\
  &\geq \alpha((1-2\lambda)x_{ik} + \lambda)\\
  &> \Big(\frac{7\lambda-2}{\lambda}\Big)\Big((1-2\lambda)\Big(\frac{4\lambda}{7\lambda-2}\Big) + \lambda\Big) \\
  &= (2-\lambda) = w^+_{ij}+w_{jk}^+ +w^-_{ik}.
	\end{align*}
	\textbf{Case 2:} $(i,k) \in E$. In this case, $(b_{ij},b_{jk},b_{ik}) = ((1-\lambda)x_{ij},(1-\lambda)x_{jk},(1-\lambda)x_{ik})$ and $(w_{ij}^+,w_{jk}^+,w_{ik}^-) = (1-\lambda,1-\lambda,0)$. Then,
	\begin{align*}
 	\alpha(b_{ij}+b_{jk}+b_{ik}) &=\alpha((1-\lambda)x_{ij} + 
(1-\lambda)x_{jk}+ (1-\lambda)x_{ik})\\
&\geq \frac{7\lambda-2}{\lambda}\left((1-\lambda)\frac{2\lambda}{7\lambda -2}\right) \\&= 2(1-\lambda) = 
 w^+_{ij}+w_{jk}^+ +w^-_{ik}.
	\end{align*}  
\end{proof}
Gleich et al.~\cite{Veldt2018ccgen} showed that for small enough $\lambda$, the canonical LambdaCC LP relaxation has an $O(\log n)$ integrality gap, but showed how to round that LP to obtain a $\frac{1}{\lambda}$-approximation, which is better than $O(\log n)$ for all $\lambda = \omega(1/\log n)$. These previous results rule out the possibility of obtaining an approximation better than $O(\log n)$ for arbitrarily small $\lambda$ by rounding the (looser) LambdaSTC LP relaxation. However, we can show that Algorithm~\ref{alg:labelcclp} will still provide a $1 + 1/\lambda$ approximation, which is very close to the approximation factor obtained by Gleich et al.~\cite{Veldt2018ccgen} for rounding a much more expensive LP.
\begin{theorem}
\label{thm:cheap}
	 When $\lambda < 1/2$, Algorithm~\ref{alg:labelcclp} is a randomized $(1 + 1/\lambda)$-approximation algorithm for LambdaCC. 
\end{theorem}
\begin{proof}
    We prove the result by showing that the conditions of Theorem~\ref{thm_2.1} are satisfied with $\alpha = (1+ \lambda)/\lambda$. Condition (1) of this theorem is easy to satisfy for a node pair $(i,j)$ if $(i,j)\in E \cap \hat{E}$, since $w_{ij}^- = 0$. It is similarly easy to satisfy if $(i,j)\notin E$ and $(i,j) \notin \hat{E}$ since then $w_{ij}^+ = 0$. The construction of $\hat{G}$ ensures it is impossible for  $(i,j)\in E$ if $(i,j)\notin \hat{E}$. If $(i,j)\notin E$ but $(i,j) \in \hat{E}$, then we know $w_{ij}^- = \lambda $ and $b_{ij} = \lambda (1-x_{ij})$ and $x_{ij} < \frac{\lambda}{1+\lambda}$. Thus,
\begin{align*}
    \alpha b_{ij} = \left(\frac{1+\lambda}{\lambda} \right)\lambda (1-x_{ij}) > (1+\lambda)\left(1- \frac{\lambda}{1+\lambda} \right) = 1 > \lambda = w_{ij}^-,
\end{align*}
which proves Condition (1) of Theorem~\ref{thm_2.1}.

Next, we prove Condition (2) for every triplet $\{i,j,k\}$ that defines an open wedge centered at $j$ in $\hat{G}$, i.e., $(i,j) \in \hat{E}$, $(j,k) \in \hat{E}$, and $(i,k) \notin \hat{E}$. We know that $(i,k) \notin E$ and $x_{ik} \geq \frac{\lambda}{1+\lambda}$, or else by the construction of $\hat{G}$ we would have $(i,k) \in \hat{E}$. Node pairs $(i,j)$ and $(j,k)$ may or may not be edges in $G$, so we separately consider 4 cases in proving $w_{ij}^+ + w_{jk}^+ + w_{ik}^- \leq \alpha(b_{ij}+b_{jk}+b_{ik})$.
    
\noindent\textbf{Case 1:} When $(i,j)\in E$ and $(j,k)\in E$, we have
$(b_{ij},b_{jk},b_{ik}) = ((1-\lambda)x_{ij},(1-\lambda)x_{jk},\lambda (1-x_{ik}))$.
The triplet $(i,j,k)$ is also an open wedge in $G$, so we have
\begin{align*}
 	\alpha(b_{ij}+b_{jk}+b_{ik})  
  &=\alpha((1-\lambda)(x_{ij}+x_{jk}) +\lambda(1-x_{ik}))\\
  &\geq \alpha((1-\lambda)x_{ik} + \lambda -\lambda x_{ik}) \\
  &= \alpha((1-2\lambda)x_{ik} + \lambda)\\
  &\geq \left(\frac{1+\lambda}{\lambda}\right)\left((1-2\lambda)\left(\frac{\lambda}{1+\lambda}\right) + \lambda\right) \\
  &= 2-\lambda = w^+_{ij}+w_{jk}^+ +w^-_{ik}.
	\end{align*}
\textbf{Case 2:} When $(i,j)\in E$ and $(j,k)\notin E$, we have
    $(b_{ij},b_{jk},b_{ik}) = ((1-\lambda)x_{ij},\lambda(1-x_{jk}),\lambda (1-x_{ik}))$
and we know $x_{jk} < \frac{\lambda}{1+\lambda} \implies (1-x_{jk}) > 1-\frac{\lambda}{1+\lambda} = \frac{1}{1+\lambda}$. Then,
\begin{align*}
\alpha(b_{ij}+b_{jk}+b_{ik})  &\geq \alpha b_{jk} = \alpha \lambda (1- x_{jk}) \geq \left(\frac{1+ \lambda}{\lambda}\right) \lambda \left(\frac{1}{1+\lambda}\right) \\
&=1 = w_{ij}^+ + w_{jk}^+ + w_{ik}^-.
\end{align*}
\textbf{Case 3:} When $(i,j)\notin E$ and $(j,k)\in E$, this is symmetric to \textbf{Case 2}.

\noindent\textbf{Case 4:} When $(i,j)\notin E$ and $(j,k)\notin E$, we have
    $(b_{ij},b_{jk},b_{ik}) = (\lambda(1-x_{ij}),\lambda(1-x_{jk}),\lambda (1-x_{ik}))$,
and both $x_{ij}$ and $x_{jk}$ are strictly less than $\frac{\lambda}{1+\lambda}$, so
\begin{align*}
     \alpha(b_{ij}+b_{jk}+b_{ik})  &=\alpha(\lambda(1-x_{ij})+\lambda(1-x_{jk})+\lambda(1-x_{ik}))\\
     &> \alpha 2\lambda \left(1-\frac{\lambda}{1+\lambda}\right)
     =2 \ge \lambda = w_{ij}^+ + w_{jk}^+ +w_{ik}^-.
\end{align*}

\end{proof}

\subsection{A 3-approximation via an intermediate LP}
\label{sec:intermediate}
Veldt et al.~\cite{veldt2018correlation,veldt2017unifying} originally presented a 3-approximation algorithm for $\lambda \ge 1/2$ based on the canonical LP relaxation. This algorithm however comes with an $O(n^3)$ size constraint matrix since all triangle inequality constraints are considered for all triplet of nodes $(i,j,k) \in {V \choose 3}$. In contrast, in the previous section, we proposed a faster 6-approximation algorithm based on the LambdaSTC LP relaxation that 
includes a triangle inequality constraint $x_{ik} \leq x_{jk} + x_{ij}$ only when $(i,j,k)$ an open wedge (centered at $j$) in $G$. In this section, we show how to obtain a 3-approximation by rounding an LP relaxation 
whose constraint set lies somewhere between the LambdaSTC and canonical LambdaCC LP relaxation.
In more detail, we include a triangle inequality constraint for every wedge in $G$ as well as for every \emph{triangle} in $G$. This is a superset of the constraint set for the LambdaSTC LP relaxation but does not include a triangle inequality constraint for every $\{i,j,k\}$. Formally, this LP relaxation is given by
\begin{equation}
	\label{eq:3cc} 
	\begin{aligned}
		\text{ min}\quad &  {\textstyle \sum_{(i,j)\in E}(1-\lambda)x_{ij} + \sum_{(i,j)\notin E}\lambda(1-x_{ij})}\\
		\text{s.t.} \quad &  x_{ij}+ x_{jk} \ge x_{ik} \text{ if } (i,j,k) \in \mathcal{W}_j \text{ or } (i,j,k) \in \mathcal{T}_j \\
		& 0\le x_{ij} \le 1 \text{ for all } (i,j) \in {\textstyle {V \choose 2}   }
	\end{aligned}
\end{equation}
where $\mathcal{T}_j$ represents a triangle that includes node $j$ as one of its vertices. Algorithm~\ref{alg:3cclp} uses the same rounding strategy that was used for the canonical LP relaxation~\cite{veldt2018correlation}, except that it is applied to the LP in equation~\eqref{eq:3cc} rather than the canonical LP. 

\begin{algorithm}[tb]
\label{alg:3cclp}
\caption{Rounding the LP~\eqref{eq:3cc} into a clustering}
\label{alg:3cclp}
	\begin{algorithmic}[1]
		\STATE{\bfseries Input:} Undirected graph $G=(V,E) ; \lambda \ge 1/2$ 
		\STATE {\bfseries Output:} Clustering of $G$.
		\STATE Solve LP~\eqref{eq:3cc} and obtain fractional $x_{ij}$ values
		\STATE Construct $\hat{G} = (V, \hat{E})$ where 
		\begin{equation*}
            \hat{E} = \{(i,j): x_{ij} < 1/3 \}           
		\end{equation*}
		\STATE Return $\textsc{Pivot}(\hat{G})$
	\end{algorithmic}
\end{algorithm}
The LP relaxation presented here has a constraint size determined by the number of wedges and triangles, denoted as $|\mathcal{W}|+|\mathcal{T}|$, in the graph. While both $|\mathcal{W}|$ and $|\mathcal{T}|$ can potentially be $O(n^3)$ in the worst case, this is not typically the scenario in practical situations. In real-world networks, the number of wedges and triangles is significantly smaller. Figure~\ref{fig:1} illustrates this observation by comparing the number of constraints in the intermediate LP~\eqref{eq:3cc} against the canonical LP.  Thus, solving and rounding this LP is more efficient compared to existing techniques, and we now prove that this can be done without a loss in the approximation factor.
\begin{theorem}
    Algorithm~\ref{alg:3cclp} is a randomized 3-approximation algorithm for LambdaCC when $\lambda\ge1/2$.
\end{theorem}
\begin{proof}
    We can prove that Algorithm~\ref{alg:3cclp} satisfies Theorem~\ref{thm_2.1} by making slight modifications to the proof presented in Theorem 6 in the work of Veldt et al.~\cite{veldt2017unifying}. Condition (1) of Theorem~\ref{thm_2.1} can be satisfied following the proof as is in~\cite{veldt2017unifying}. To prove Condition (2), we demonstrate that 
    \begin{equation}
        \label{ijkcon}
        w_{ij}^++w_{jk}^++w_{ik}^- \le \alpha(b_{ij}+b_{jk}+b_{ik})
    \end{equation} for every triplet of nodes $(i,j,k)$ that is mapped to an open wedge centered at $j$ in $\hat{G}$. This means that $x_{ij},x_{jk} < 1/3$, $x_{ik}\ge 1/3$.
    Note that we are only able to apply the triangle inequality $x_{ik}\le x_{ij}+x_{jk}$ if $(i,j,k)$ is also an open wedge or a triangle in the original graph $G$.
    Given an arbitrary triplet $(i,j,k)$ that maps to an open wedge in $\hat{G}$, there are 8 possibilities for the edge structure in $G$, depending on which pairs of nodes in $(i,j,k)$ share an edge in $G$. Following Veldt et al.~\cite{veldt2017unifying}, we can prove inequality~\eqref{ijkcon} for each of the 8 cases separately. 
    Note that we do not need to update the analysis for cases where $(i,j,k)$ is an open wedge or a triangle in $G$, since our new LP in~\eqref{eq:3cc} includes triangle inequality constraints for these cases. 
    This means that the following cases from the analysis of Veldt et al.~\cite{veldt2017unifying} remain unchanged:
    \begin{itemize}
        \item \textbf{Case 1}: $(i,j,k)$ forms a wedge centered at $j$ in $G$.
        \item \textbf{Case 5} and \textbf{Case 6:} $(i,j,k)$ forms a wedge centered at either $i$ or $k$.
        \item \textbf{Case 8:} $(i,j,k)$ forms a triangle.
    \end{itemize}
For \textbf{Case 7}, where $(i,k) \in E$, $(i,j)\notin E$, and $(j,k) \notin E$, the proof is trivial since $w_{ij}^+ +w_{jk}^+ +w_{ik}^- = 0$. We update the proof for the remaining cases as follows:

    \noindent\textbf{Case 2:} When $(i,j)\in E,(j,k)\notin E,(i,k)\notin E$, we have $(b_{ij},b_{jk},b_{ik})=((1-\lambda)x_{ij},\lambda(1-x_{jk}),\lambda(1-x_{ik}))$ and $(w_{ij}^+,w_{jk}^+,w_{ik}^-) = (1-\lambda,0,\lambda)$. Thus,
    \begin{align*}
        \alpha(b_{ij}+b_{jk}+b_{ik}) &= \alpha((1-\lambda)x_{ij}+\lambda(1-x_{jk})+\lambda(1-x_{ik}))\\
        &\ge 3(\lambda-\lambda x_{jk}) > 3(\lambda - \lambda/3)\\
        &=2\lambda \ge 1 = w_{ij}^++w_{jk}^++w_{ik}^-.
    \end{align*}
    \noindent\textbf{Case 3:} When $(i,j)\notin E,(j,k)\in E,(i,k)\notin E$ is symmetric to \textbf{Case 2} and the same result holds.
     
    \noindent\textbf{Case 4:} When $(i,j)\notin E,(j,k)\notin E,(i,k) \notin E$, we have 
    $(b_{ij},b_{jk},b_{ik})=(\lambda(1-x_{ij}),\lambda(1-x_{jk}),\lambda(1-x_{ik}))$ and $(w_{ij}^+,w_{jk}^+,w_{ik}^-) = (0,0,\lambda)$. Then,
    \begin{align*}
        \alpha(b_{ij}+b_{jk}+b_{ik}) &= \alpha(\lambda(1-x_{ij})+\lambda(1-x_{jk})+\lambda(1-x_{ik}))\\
        &>3\lambda(2/3 + 2/3)\\ 
        &= 3\lambda(4/3) = 4\lambda > \lambda = w_{ij}^++w_{jk}^++w_{ik}^-.
    \end{align*}
\end{proof}
Therefore, considering all the cases, we can conclude that Theorem~\ref{thm_2.1} holds for $\alpha = 3$, satisfying the 3-approximation guarantee.

\subsection{Runtime Analysis}
For a graph $G = (V,E)$, let $m = |E|$ and $n = |V|$. When written in the form $\min_{\textbf{A}\textbf{x} = \textbf{b}} \textbf{c}^T \textbf{x}$, the canonical LP relaxation for LambdaCC has $O(n^3)$ constraints and variables. Even using recent theoretical algorithms for solving linear programs in matrix multiplication time~\cite{cohen2021solving,jiang2021faster}, the runtime is $\Omega((n^3)^\omega)$ where $\omega$ is the matrix multiplication exponent, so the runtime for solving and rounding the canonical relaxation is $\Omega(n^6)$. Not only does this have a prohibitively expensive runtime, but in practice even forming such a large constraint matrix can lead to memory issues that make it infeasible to solve this on a very large scale. Thus, although this approach provides the best theoretical approximation factor, it is not scalable. 

Our new approximation algorithms come with good approximation guarantees and are much faster than solving the canonical relaxation, both in theory and practice. Finding the open wedges of $G$ can be done in time $O(\sum_{v \in V} d_v^2)$ by visiting each node, and then visiting each pair of neighbors of that node in turn. This runtime is upper bounded by $O(mn)$. When applying the randomized \textsc{Pivot} algorithm, this is in fact the most expensive part of CFP, so the overall runtime for CFP is $O(\sum_{v \in V} d_v^2) = O(nm)$. If we use the deterministic pivoting strategy of van Zuylen and Williamson~\cite{vanzuylen2009deterministic}, this can be implemented in $O(n^3)$ time so that is the runtime for a derandomized version of CFP.

The LambdaSTC LP is a covering LP, so for $\varepsilon \geq 0$ we can find a $(1+\varepsilon)$-approximate solution in $\tilde{O}(\frac{1}{\varepsilon^2}|\mathcal{W}|)$ time using the multiplicative weights update method~\cite{quanrud2020nearly,garg2004fractional,fleischer2004fast}, where $\tilde{O}$ suppresses logarithmic factors. This assumes we already know $|\mathcal{W}|$; if we factor in the time it takes to find all open wedges the runtime comes to $\tilde{O}(\frac{1}{\varepsilon^2}|\mathcal{W}| + \sum_{v \in V} d_v^2)$. Minor alteration to our analysis quickly shows that a $(1+\varepsilon)$-approximate solution to the LP translates to approximation factors that are a factor $(1+\varepsilon)$ larger. Once again, applying the deterministic pivot selection adds $O(n^3)$ to the runtime, which is still far better than $\Omega(n^6)$.

%% file: arxiv/sections/experiments.tex
\section{Experiments}
This section presents a performance of our algorithms.
%
%
We conduct experiments on publicly available datasets from various domains, including the SNAP~\cite{snap} and Facebook100~\cite{facebook} datasets, which are available at the Suitsparse matrix collection~\cite{suitesparse2011davis}. To implement the algorithms, we use the Julia programming language, and we run all experiments on a Dell XPS machine with 16 GB RAM and an Intel Core i7 processor. Both the canonical and the LambdaSTC LP relaxations are solved using Gurobi optimization software~\cite{gurobi2021gurobi}. We focus here on finding exact solutions for the LambdaSTC LP relaxation using existing optimization software. This is already far more scalable than trying to form the constraint matrix for the canonical LP relaxation and solve it using Gurobi. Finding faster approximate solutions for the LP using the multiplicative weights update method is a promising direction for future research, but is beyond the scope of the current paper. Code for our implementations and experiments is available at~\url{https://github.com/Vedangi/FastLamCC}.

\subsection{Approximation algorithms for LambdaCC}
A natural question to ask is how well our approximation algorithms compare against previous algorithms for LambdaCC based on the canonical LP relaxation. It is worth noting first of all that even forming the full constraint matrix for the canonical LP (which has $n(n-1)(n-2)/2 = O(n^3)$ triangle inequality constraints) becomes infeasible for even modest-sized graphs due to memory constraints. Meanwhile, the LambdaSTC LP relaxation has one triangle inequality constraint for each open wedge. Although there exist graphs such that $|\mathcal{W}| = O(n^3)$, this is not the case in practical situations. Figure~\ref{fig:1} plots the size of $|\mathcal{W}|$ for all of the Facebook100 networks, as well as for a range of graphs of different classes from SNAP (e.g., social networks, citation networks, web networks, etc.). In all cases $|\mathcal{W}|$ is orders of magnitude smaller than $n(n-1)(n-2)/2$, illustrating that solving and rounding this LP is far more practical than using existing LP-based techniques. We also plot the number of constraints in the intermediate LP relaxation from Section~\ref{sec:intermediate}, showing that it has only a slight increase in constraint size over the LambdaSTC LP.

An additional reason to use the LambdaSTC relaxation is that in practice, \textit{solving the LambdaSTC relaxation often also solves the canonical LP relaxation.} This can be checked by seeing whether the optimal LP variables for the LambdaSTC relaxation are also feasible for the canonical LP.\footnote{This can also be viewed as the first step in a more memory efficient approach for solving correlation clustering LP relaxations that has been applied elsewhere~\cite{veldt2019metric,veldt2022stc}: solve the LP over a subset of the constraints and iteratively add more constraints until the variables satisfy all triangle inequalities. Our results indicate that for these graphs and $\lambda$ values, enforcing triangle inequality constraints just at open wedge is sufficient.}
Table~\ref{tab:3graphs} shows results for solving and rounding the LambdaSTC LP (Algorithm~\ref{alg:labelcclp}) on three graphs for a range of different $\lambda$ values. The graphs are Simmons81 (a social network), ca-GrQc (a collaboration network), and Polblogs (a Political blogs network). We attempted to form the full canonical LP relaxation for these graphs and solve it but quickly ran out of memory. We were able to form and solve the LambdaSTC LP relaxation, and in almost all cases the optimal solution variables for this LP were certified as being feasible (and hence optimal) for the canonical LP. Thus, our LP-based algorithm far exceeded its theoretical guarantees. In practice, it produced solutions that are within a factor of 2 or less from the LP lower bound. When rounding, we applied both our new approach (Algorithm~\ref{alg:labelcclp}) as well as the existing rounding strategy for the canonical LP relaxation, since the rounding step is very fast. We used the previous rounding strategy for the canonical LP whenever we could certify we had solved the canonical LP (since this has an improved a priori guarantee). In practice though, the results for the two different rounding strategies were nearly indistinguishable.

Table~\ref{tab:3graphs} also displays results for CFP, showing that it is even orders of magnitude faster than solving and rounding the LambdaSTC LP relaxation while producing comparable approximation ratios (ratio between clustering solution and the computed lower bound). While solving our LP relaxation takes up to hundreds of seconds on the three graphs, CFP takes mere fractions of a second.

\begin{table*}[t]
    \centering
    \caption{Results for CFP and rounding the LambdaSTC LP relaxation on three graphs. An asterisk indicates when solving the LambdaSTC relaxation produced the optimal solution for the canonical LP.}
    \label{tab:3graphs}
    \begin{tabular}{l l  l l  l l  l l  l l  } 
        \hline
        \multicolumn{2}{c}{} & \multicolumn{2}{c}{Lower Bound} & \multicolumn{2}{c}{Clustering score} & \multicolumn{2}{c}{Ratio} & \multicolumn{2}{c}{Runtime (seconds)}\\
        \cmidrule(lr){3-4}  \cmidrule(lr){5-6}   \cmidrule(lr){7-8} \cmidrule(lr){9-10}
        \emph{Graph} & $\lambda$ & CFP & LambdaSTC & CFP & LambdaSTC & CFP & LambdaSTC & CFP & LamdaSTC \\  
        \hline
            & 0.4 & 2668 & 2889.7 & 6043 $\tiny{\pm 68}$
        & 4611 $\tiny{\pm 47}$ & 2.3 $\tiny{\pm 0.026}$
        & 1.6 $\tiny{\pm 0.016}$ & 0.34 $\tiny{\pm 0.096}$ &  38.3 $\tiny{\pm 0.018}$ \\
         ca-GrQc   & 0.55 & 2064 & 2236.5$^*$ & 4092 $\tiny{\pm 33}$
        & 3708 $\tiny{\pm 0}$ & 2.0 $\tiny{\pm 0.016}$
        & 1.7 $\tiny{\pm 0}$ & 0.061 $\tiny{\pm 0.0084}$ &  29.7 $\tiny{\pm 0.0049}$ \\
         n = 5242   & 0.75 & 1179 & 1278.2 & 2373 $\tiny{\pm 20}$
        & 2118 $\tiny{\pm 1}$ & 2.0 $\tiny{\pm 0.017}$
        & 1.7 $\tiny{\pm 0.0004}$ & 0.058 $\tiny{\pm 0.0079}$ &  27.7 $\tiny{\pm 0.006}$ \\  

         m = 14484   & 0.95 & 239 & 259.3 & 469 $\tiny{\pm 3}$
        & 430 $\tiny{\pm 0}$ & 2.0 $\tiny{\pm 0.011}$
        & 1.7 $\tiny{\pm 0}$ & 0.055 $\tiny{\pm 0.0083}$ &  25.4 $\tiny{\pm 0.0079}$ \\
        \midrule
            & 0.4 & 9823 & 9893.8$^*$ & 21569 $\tiny{\pm 72}$
        & 20674 $\tiny{\pm 110}$ & 2.2 $\tiny{\pm 0.0073}$
        & 2.1 $\tiny{\pm 0.011}$ & 0.48 $\tiny{\pm 0.1}$ &  3064.3 $\tiny{\pm 0.028}$ \\
         Simmons81     & 0.55 & 7392 & 7420.5$^*$ & 15797 $\tiny{\pm 52}$
        & 14839 $\tiny{\pm 0}$ & 2.1 $\tiny{\pm 0.0071}$
        & 2.0 $\tiny{\pm 0.0}$ & 0.25 $\tiny{\pm 0.0075}$ &  2935.4 $\tiny{\pm 0.0067}$ \\   
          n = 1518 & 0.75 & 4113 & 4122.5$^*$ & 8657 $\tiny{\pm 18}$
        & 8244 $\tiny{\pm 0}$ & 2.1 $\tiny{\pm 0.0044}$
        & 2.0 $\tiny{\pm 0.0}$ & 0.12 $\tiny{\pm 0.007}$ &  619.7 $\tiny{\pm 0.0023}$ \\     
         m = 32988    & 0.95 & 822 & 824.5$^*$ & 1646 $\tiny{\pm 0}$
        & 1649 $\tiny{\pm 0}$ & 2.0 $\tiny{\pm 0.0005}$
        & 2.0 $\tiny{\pm 0}$ & 0.098 $\tiny{\pm 0.0065}$ &  464.8 $\tiny{\pm 0.0028}$ \\
\midrule
            & 0.4 & 4960 & 5013.1$^*$ & 10591 $\tiny{\pm 59}$
        & 9997 $\tiny{\pm 120}$ & 2.1 $\tiny{\pm 0.012}$
        & 2.0 $\tiny{\pm 0.024}$ & 0.49 $\tiny{\pm 0.13}$ &  244.4 $\tiny{\pm 0.0018}$ \\
        
          Polblogs    & 0.55 & 3745 & 3760.2$^*$ & 7883 $\tiny{\pm 23}$
        & 7517 $\tiny{\pm 0}$ & 2.1 $\tiny{\pm 0.0062}$
        & 2.0 $\tiny{\pm 0.0}$ & 0.21 $\tiny{\pm 0.0074}$ &  217.4 $\tiny{\pm 0.00035}$ \\    
        n = 1222    & 0.75 & 2084 & 2089.0$^*$ & 4377 $\tiny{\pm 16}$
        & 4177 $\tiny{\pm 0}$ & 2.1 $\tiny{\pm 0.0075}$
        & 2.0 $\tiny{\pm 0.0}$ & 0.071 $\tiny{\pm 0.01}$ &  187.9 $\tiny{\pm 0.0079}$ \\      

        m = 16714    & 0.95 & 417 & 417.8$^*$ & 837 $\tiny{\pm 0}$
        & 835 $\tiny{\pm 0}$ & 2.0 $\tiny{\pm 0}$
        & 2.0 $\tiny{\pm 0}$ & 0.052 $\tiny{\pm 0.0089}$ &  114.2 $\tiny{\pm 0.0024}$ \\
        \hline
    \end{tabular}
    \label{tab:summary} 
\end{table*}

\begin{figure}
    \begin{minipage}[t]{.495\linewidth}
        \centering
        \includegraphics[width=\linewidth]{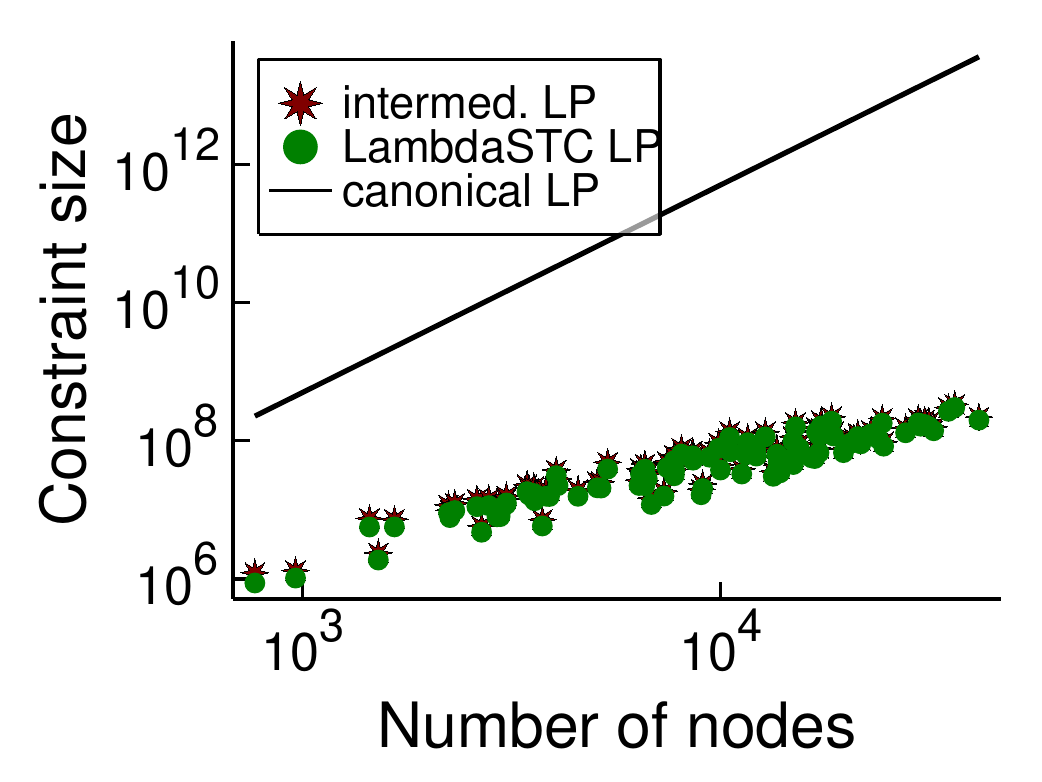}
        \label{fig:sub31}
    \end{minipage}\hfill
    \begin{minipage}[t]{.495\linewidth}
        \centering
        \includegraphics[width=\linewidth]{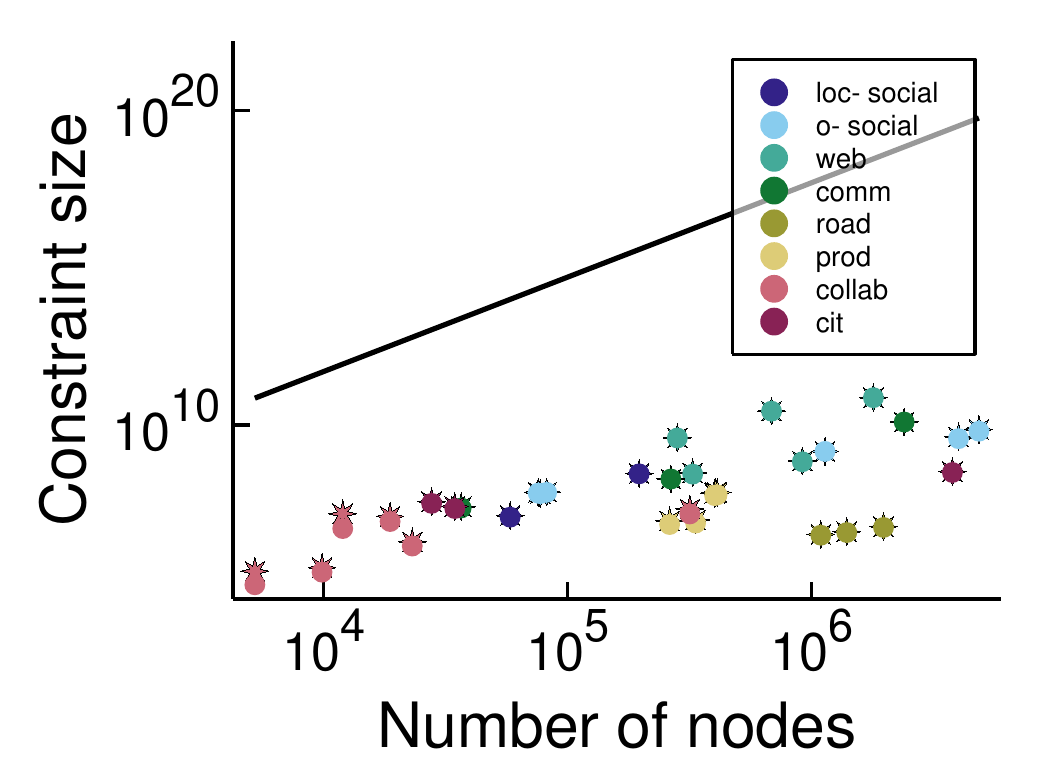}
        \label{fig:sub32}
    \end{minipage}
    \vspace{-\baselineskip}
    \caption{
    Comparing number of constraints in the canonical LP against the number of constraints in the LambdaSTC LP and the intermediate LP (from Section~\ref{sec:intermediate}) for graphs from Facebook100 (left) and SNAP datasets (right). Each dot represents the number of open wedge constraints for each graph while each star represents the number of constraints for the intermediate LP. We use the same SNAP graphs as Veldt~\cite{veldt2022stc}, and color code based on their type (location-based social, other social, web, communication, road, product, collaboration, and citation networks). }
    \label{fig:1}
\end{figure}

\subsection{Combining CFP with Fast Heuristics}
The Louvain method is a widely used heuristic clustering technique that greedily moves nodes in order to optimize a clustering objective~\cite{blondel2008fast}. The original Louvain method was designed for maximum modularity clustering, but many variations of the method have been designed. This includes a fast heuristic called LambdaLouvain~\cite{veldt2018correlation}, that greedily optimizes the LambdaCC objective for a given parameter $\lambda$, as well as a parallel version of this method~\cite{shi2021scalable}. Although these methods are fast and perform well in practice, they do not compute lower bounds for the LambdaCC objective nor provide any approximation guarantees. One benefit of our algorithms is that they come with lower bounds that can be used not only to design faster approximation algorithms for LambdaCC, but also to obtain a posteriori guarantees for other heuristic methods.

Figures~\ref{fig:21} and ~\ref{fig:22} showcase the combined results of LambdaLouvain with CFP lower bounds on graphs even larger than those considered in Table~\ref{tab:3graphs}. These results demonstrate superior a posteriori approximation ratios (clustering objective divided by CFP lower bound) compared to running CFP by itself. Notably, as $\lambda \rightarrow 1$, the difference in approximation factors between CFP and LambdaLouvain decreases, converging toward a similar outcome. We execute both the CFP rounding procedures and LambdaLouvain method for 15 iterations, reporting the mean and standard deviation for approximation ratios and runtimes. 
While the CFP+LambdaLouvain yields better approximations, it comes with longer runtimes compared to the CFP alone. Even for small values of $\lambda$, we can certify that LambdaLouvain can produce a respectable factor of around 2 by using the lower bounds generated by CFP. 

\begin{figure}
    \centering
    \subfloat{\includegraphics[width=0.45\linewidth]{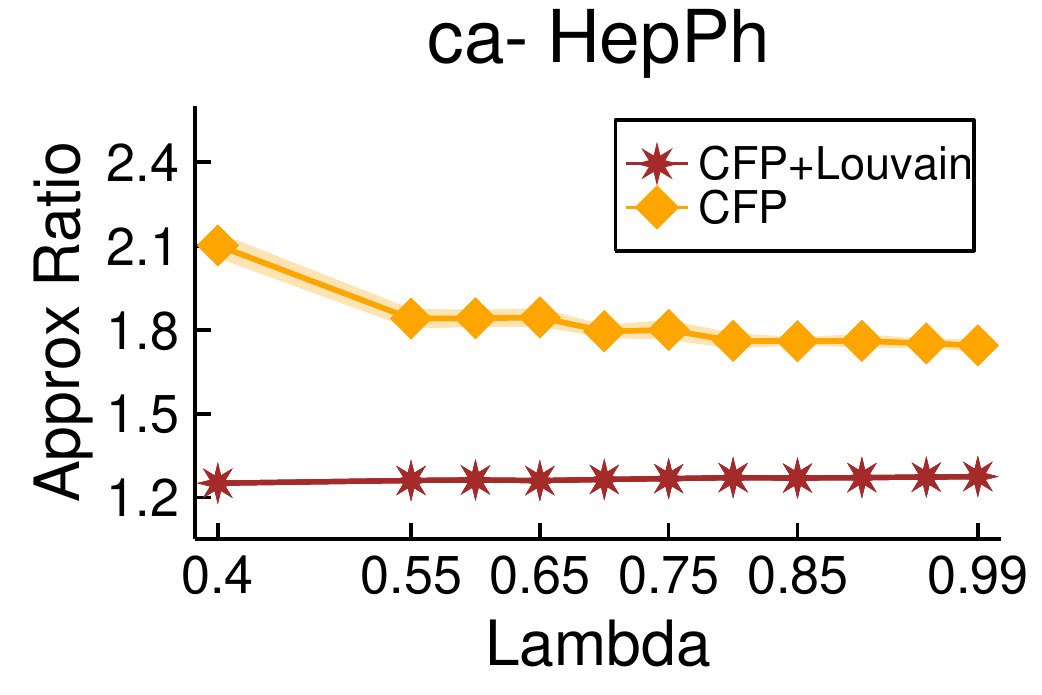}\label{fig:sub11}}\hfill
    \subfloat{\includegraphics[width=0.45\linewidth]{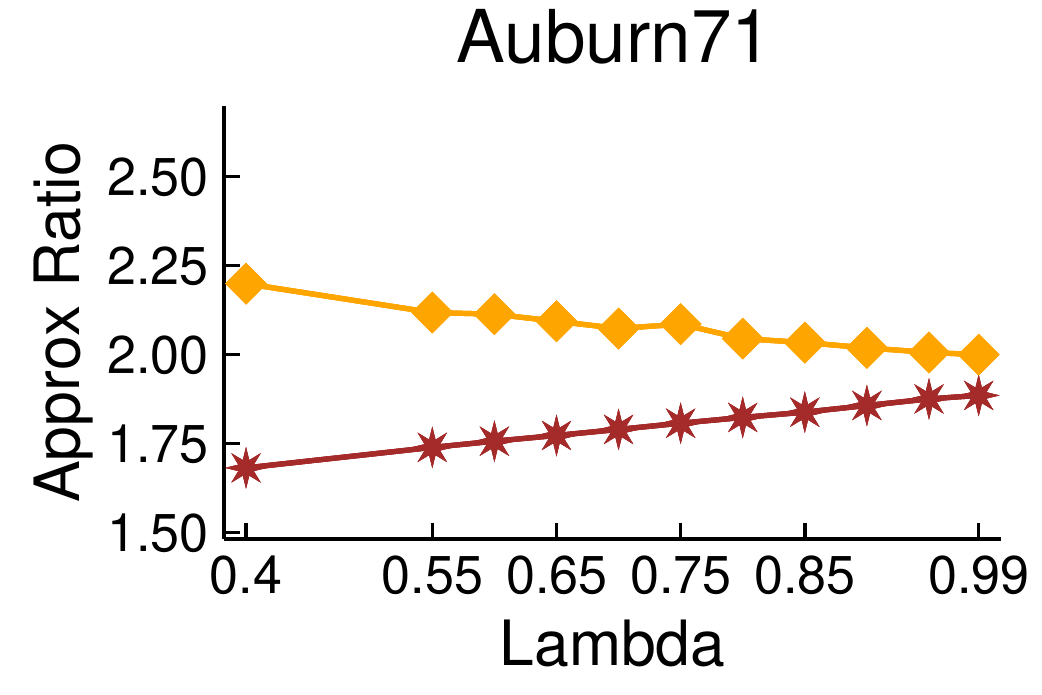}\label{fig:sub12}}
    
    \subfloat 
    {\includegraphics[width=0.45\linewidth]{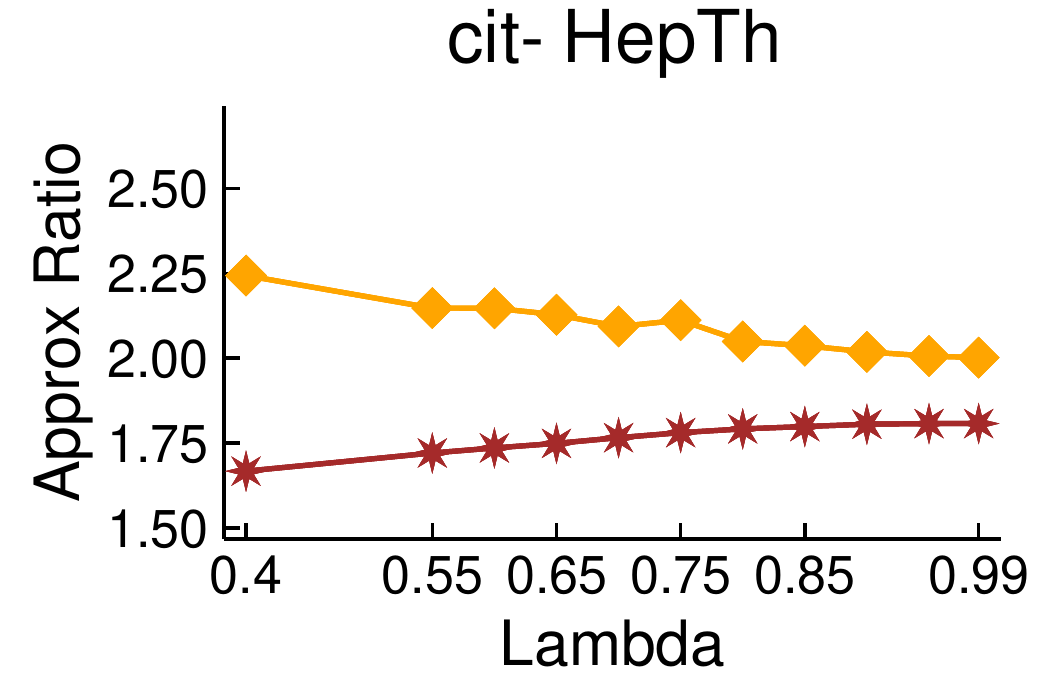}\label{fig:sub13}}\hfill
    \subfloat 
    {\includegraphics[width=0.45\linewidth]{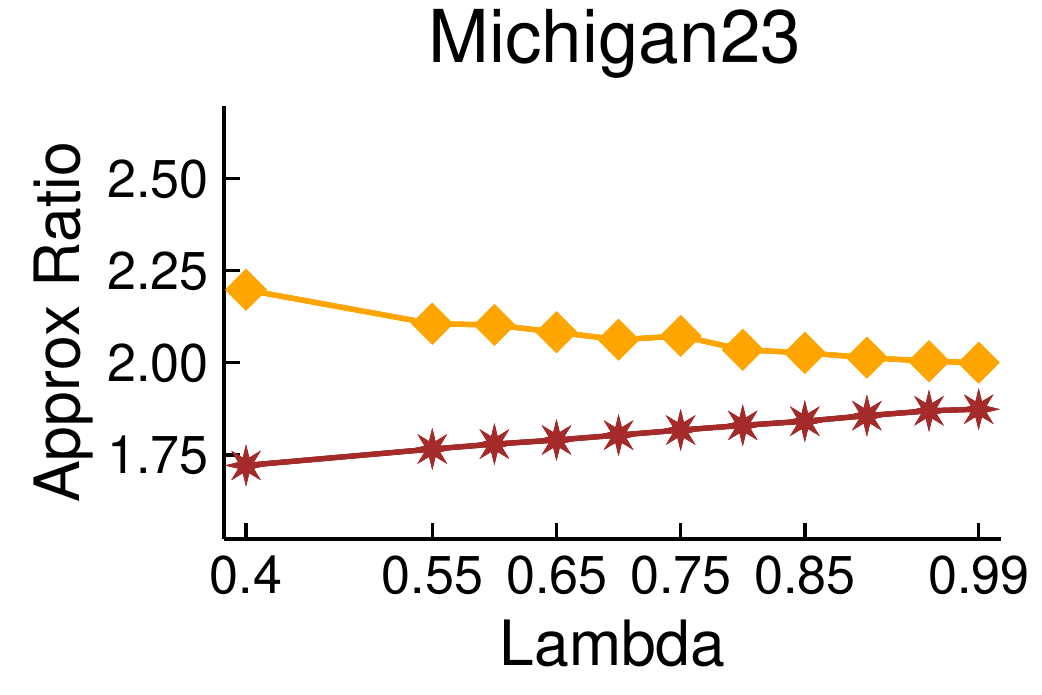}\label{fig:sub14}}
    \captionsetup[subfigure]{labelformat=empty} 
    
    \caption{Approximation ratios for combining CFP and LambdaLouvain on ca-HepPh (12K nodes, 118.5K edges), cit-HepTh(27K nodes, 352K edges) SNAP networks, and Auburn71 (18.4K nodes, 973.9K edges), Michigan23(20K nodes, 1.17M edges) facebook graphs for different values of $\lambda$. }
    \label{fig:21}
\end{figure}

\begin{figure}
    \centering
    \subfloat{\includegraphics[width=0.45\linewidth]{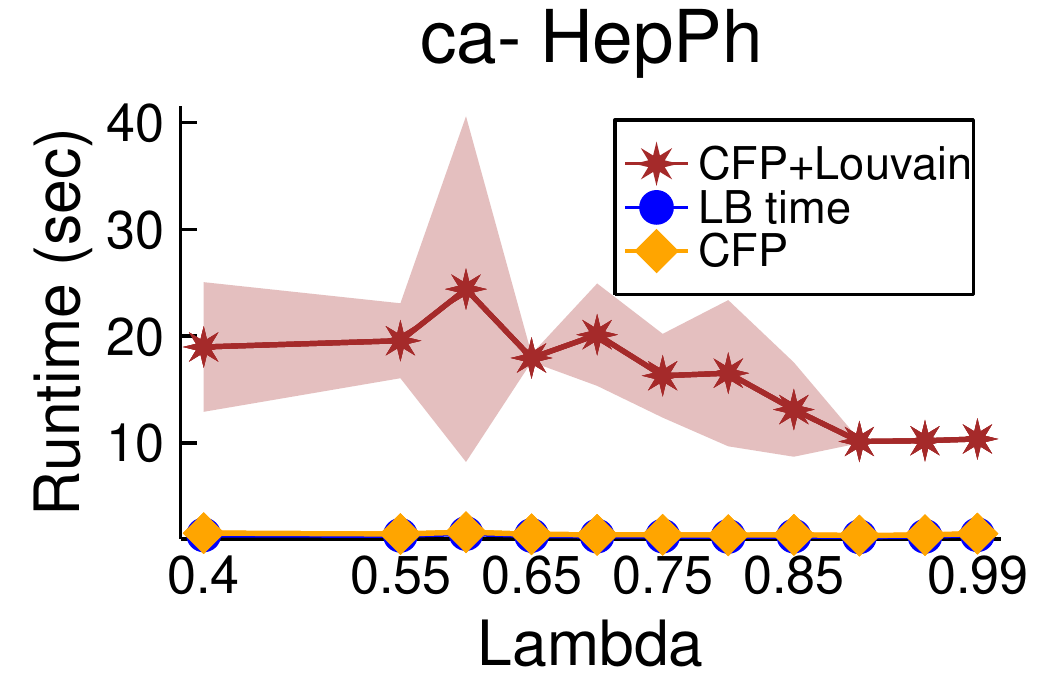}\label{fig:sub11}}\hfill
    \subfloat{\includegraphics[width=0.45\linewidth]{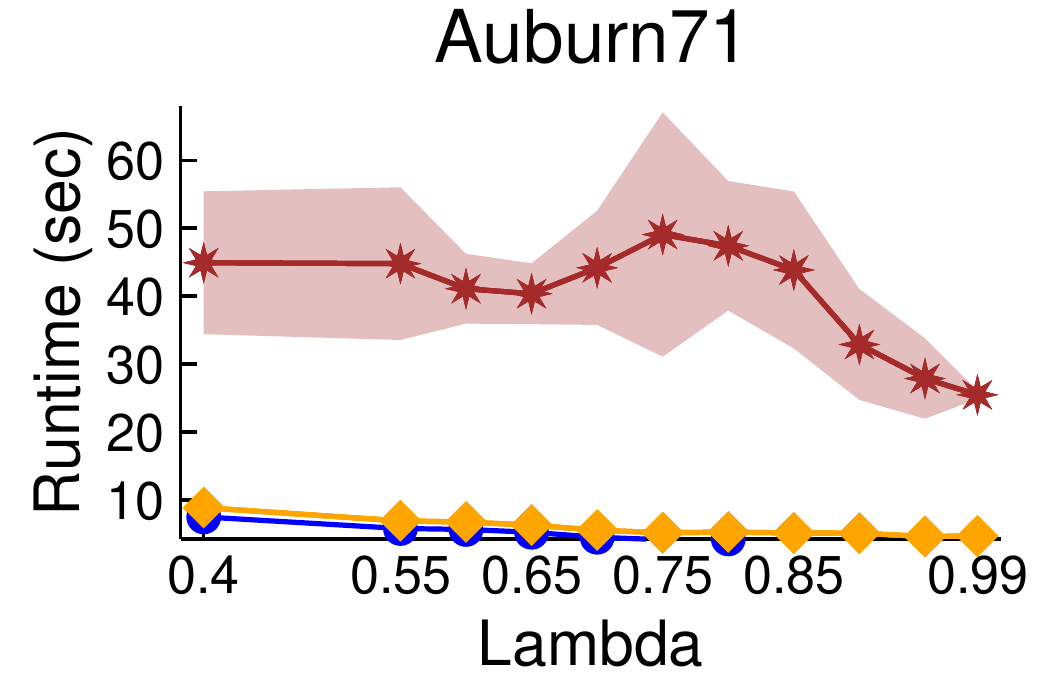}\label{fig:sub12}}  

    \subfloat 
    {\includegraphics[width=0.45\linewidth]{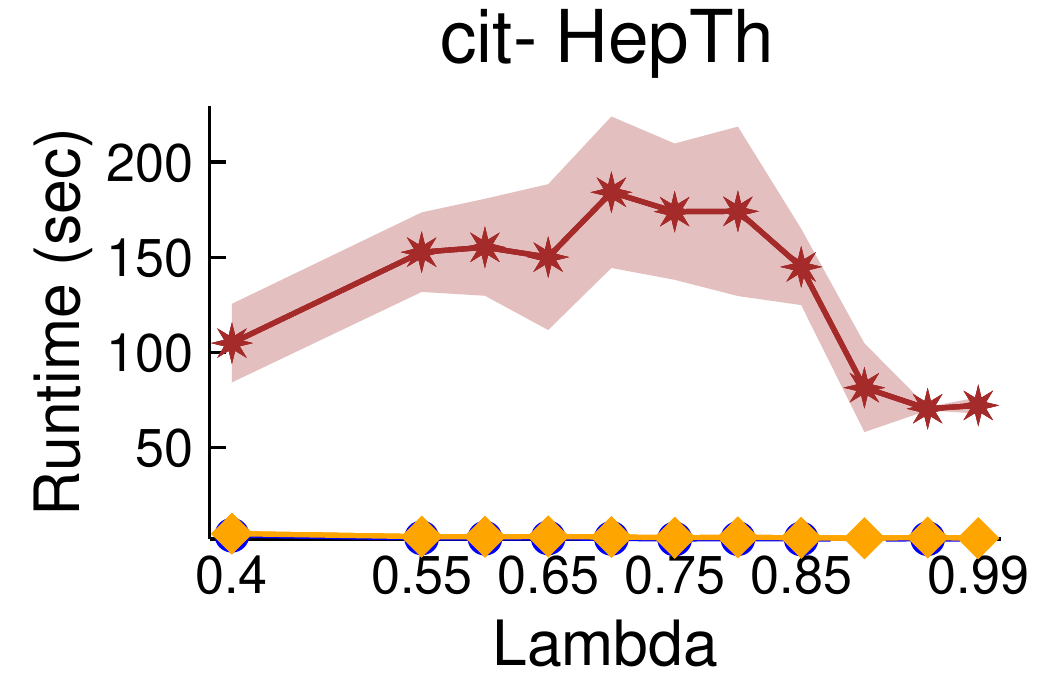}\label{fig:sub13}}\hfill
    \subfloat 
    {\includegraphics[width=0.45\linewidth]{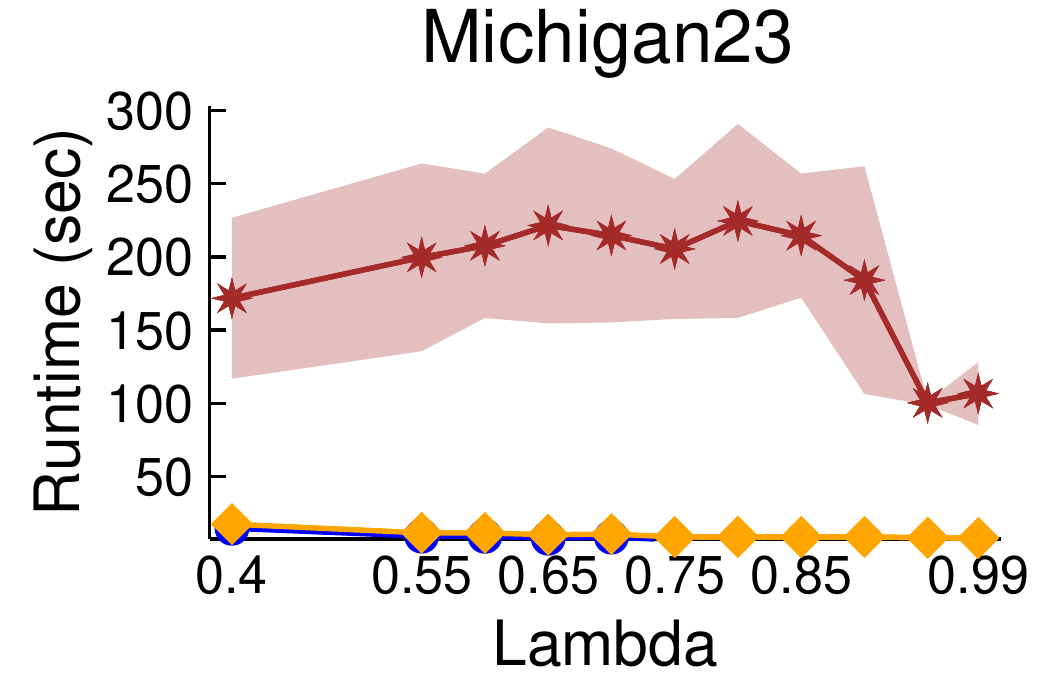}\label{fig:sub14}}
    \captionsetup[subfigure]{labelformat=empty} 
    
    \caption{Runtimes for combining CFP and LambdaLouvain on ca-HepPh, Auburn71, cit-HepTh and Michigan23 for different values of $\lambda$. The time to compute the CFP lower bound is in blue, which is the bottleneck for this algorithm.}
    \label{fig:22}
\end{figure}

\subsection{Scalability of CoverFlipPivot}
\begin{figure}
    \begin{minipage}[t]{.495\linewidth}
        \centering
        \includegraphics[width=\linewidth]{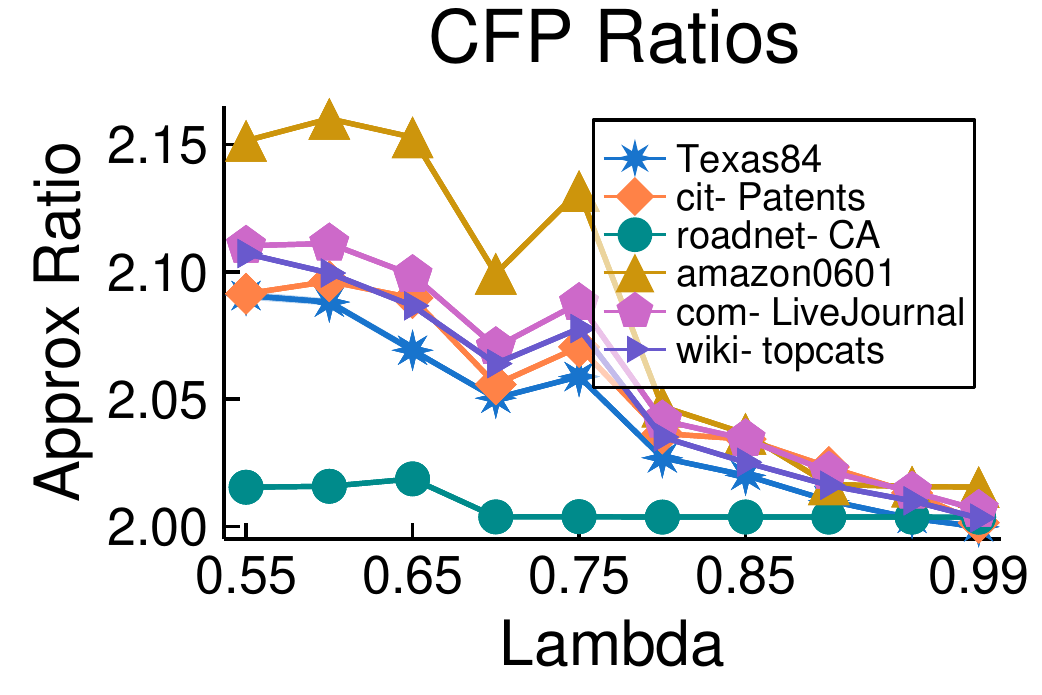}
        \label{fig:sub31}
    \end{minipage}\hfill
    \begin{minipage}[t]{.495\linewidth}
        \centering
        \includegraphics[width=\linewidth]{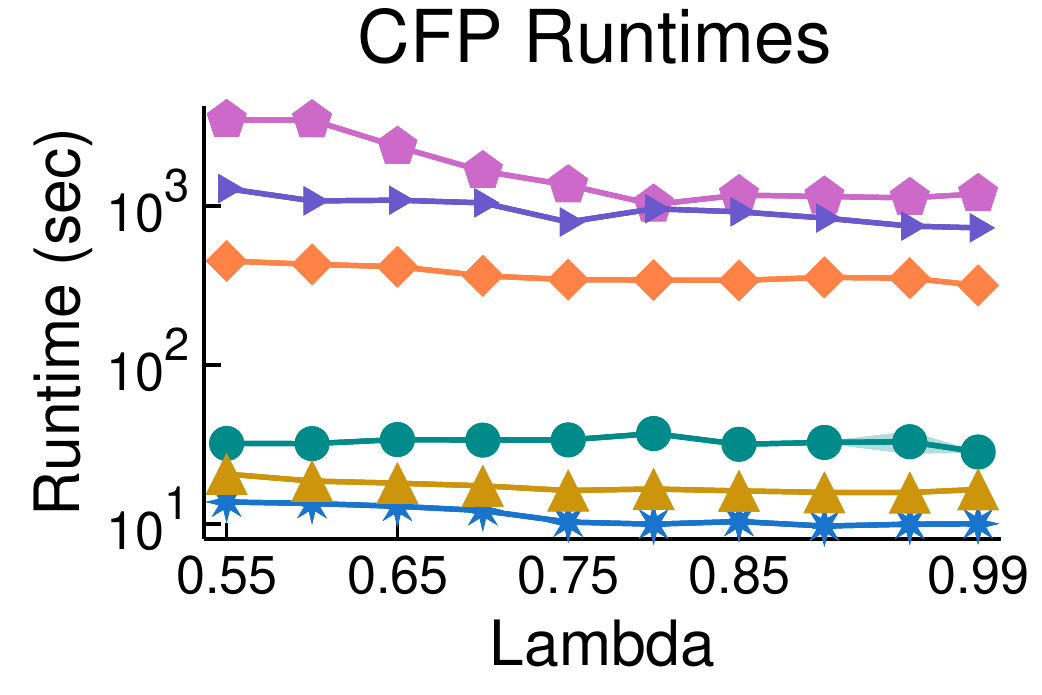}
        \label{fig:sub32}
    \end{minipage}
    \vspace{-\baselineskip}
    \caption{CFP approximations and runtimes on Texas84 (36K nodes and 1.59M edges), cit-Patents (3.6M nodes and 1.65B edges), roadNet-CA (1.97M nodes and 2.76M edges), amazon0601 (403K nodes and 2.44M edges), com-LiveJournal (3.9M nodes, 3.46B edges) and wiki-topcats (1.79M nodes and 2.54B edges).}
    \label{fig:3}
\end{figure}

We further test the limits of CFP by running it on much larger graphs. Figure~\ref{fig:3} shows approximation results on a social network with 1.59 million edges (Texas84), a road network with 2.76 million edges (roadNet-CA), a citation network with 1.65 \emph{billion} edges (cit-Patents), an Amazon product co-purchasing network with 2.4 million edges (amazon0601), Wikipedia web network with 2.54 billion edges (wiki-topcats) and a blogging community network with 3.46 billion edges (com-Journal).
CFP consistently outperforms its theoretical 6-approximation guarantee. For $\lambda \geq 0.55$, it produces approximations of 2.1 or better. When $\lambda = 0.5$, approximation factors increase to between 2.4-2.8, which still outperforms the 6-approximation guarantee. (We omit these results from the plot in Figure~\ref{fig:3} in order to zoom in and better display factors near 2 for $\lambda \geq 1/2$.)  For each value of $\lambda$, the method takes around 58 minutes for the larger com-LiveJournal graph (with 3.46B edges). In contrast, the cit-Patents graph, consisting of 1.65 billion edges, is processed in less than 11 minutes, showcasing a notably faster runtime. Furthermore, the method exhibits even quicker processing times for the other graphs.
An intriguing observation in this context is that as the objective transitions from cluster editing to cluster deletion ($\lambda\rightarrow 1$), both the approximation factor and the runtime exhibit improvements.